\newcommand{\tAppA}{0}
\newcommand{\tAppB}{0}

\newif\ifnotes
\documentclass[journal]{IEEEtran}

\usepackage{flushend}
\usepackage{graphicx}
\usepackage[english]{babel}
\usepackage{amsthm}
\usepackage{amsmath}
\usepackage[usenames,dvipsnames,table]{xcolor}
\usepackage{amssymb}
\usepackage{inputenc}
\usepackage{xspace}

\ifnotes
\usepackage[nomargin,inline,draft]{fixme}
\else
\usepackage[nomargin,inline,final]{fixme}
\fi
\usepackage[obeyFinal]{todonotes}
\usepackage[
  colorlinks=true,
  linkcolor=blue,
  citecolor=red
]{hyperref}
\usepackage[
  sorting=none,
  style=numeric-comp,
  backend=biber,
  bibencoding=utf8,
  isbn=false,
  url=true,
  doi=false,
  url=false,
  mincrossrefs=100,
  maxnames=12,
  firstinits=true
]{biblatex}

\usepackage{algorithm}
\usepackage{algpseudocode}%
\algblockdefx{Event}{EndEvent}[1]%
{\textbf{upon event }#1 \textbf{do }}%
{\textbf{}}

\algblockdefx{Function}{EndFunction}[1]%
{\textbf{Function }#1 }%
{\textbf{}}

\newcommand{\pair}[2]{\langle#1,#2\rangle}

\newcommand{\mkmsg}[2]{\langle#1,#2\rangle}
\newcommand{\mkmsgv}[3]{\langle#1,#2,#3\rangle}
\newcommand{\mkmsgvs}[4]{\langle#1,#2,#3,#4\rangle}
\newcommand{\META}[1]{{\mathit{#1}}}

\newcommand{\roundbased}{{\mathcal{R}}\mathit{ound}{\mathcal{B}}\mathit{ased}\xspace}

\newcommand{\duration}{\mathbb{T}}

\newcommand{\hb}[2]{\mbox{\texttt{HB}}\left(#1,#2\right)}
\newcommand{\rhb}{\mathcal{R}_{\mathit{HB}}}

\newcommand{\sq}{\META{sq}}

\newcommand{\bcastOpSYMB}{\mbox{\texttt{RTBRB-broadcast}}}
\newcommand{\bcastOpM}{\bcastOpSYMB()}
\newcommand{\bcastOp}[3]{\bcastOpSYMB(#1,#2,#3)}

\newcommand{\delivSYMB}{\mbox{\texttt{Deliver}}}

\newcommand{\deliv}[2]{\delivSYMB\left(#1,#2\right)}
\newcommand{\delivOpSYMB}{\mbox{\texttt{RTBRB-deliver}}}
\newcommand{\delivOpM}{\delivOpSYMB()}
\newcommand{\delivOp}[3]{\delivOpSYMB(#1,#2,#3)}
\newcommand{\rdelivSYMB}{\mathcal{R}_{\mathit{deliver}}}
\newcommand{\rdeliv}[3]{\rdelivSYMB(#1,#2,#3)}

\newcommand{\delivMode}{\texttt{deliver}}

\newcommand{\echoSYMB}{\mbox{\texttt{Echo}}}
\newcommand{\echoM}{\echoSYMB\left(\right)}

\newcommand{\echo}[2]{\echoSYMB\left(#1,#2\right)}
\newcommand{\rechoSYMB}{\mathcal{R}_{\mathit{echo}}}
\newcommand{\recho}[3]{\rechoSYMB(#1,#2,#3)}

\newcommand{\echoMode}{\texttt{echo}}

\newcommand{\deliverMsgSYMB}{\texttt{deliver-msg}}
\newcommand{\deliverMsg}[4]{\deliverMsgSYMB(#1,#2,#3,#4)}
\newcommand{\deliverMsgAt}[5]{\deliverMsgSYMB_{#1}(#2,#3,#4,#5)}

\newcommand{\tDiffuseSYMB}{\texttt{b-diffuse}}
\newcommand{\tDiffuseF}{\tDiffuseSYMB()}
\newcommand{\tDiffuse}[3]{\tDiffuseSYMB(#1,#2,#3)}

\newcommand{\initializationSYMB}{\texttt{initialization}}
\newcommand{\initialization}{\initializationSYMB()}

\newcommand{\checkConnectivitySYMB}{\texttt{check-connectivity}}
\newcommand{\checkConnectivity}{\checkConnectivitySYMB()}

\newcommand{\timeOutSYMB}{\texttt{Timeout}}
\newcommand{\timeOut}[3]{\timeOutSYMB(#1,#2,#3)}
\newcommand{\expiredTimeoutSYMB}{\texttt{Expired-Timer}}
\newcommand{\expiredTimeout}[3]{\expiredTimeoutSYMB(#1,#2,#3)}

\newcommand{\timeOutConnSYMB}{\texttt{Timeout}}
\newcommand{\timeOutConn}[2]{\timeOutConnSYMB(#1,#2)}
\newcommand{\expiredTimeoutConnSYMB}{\texttt{Expired-Timer}}
\newcommand{\expiredTimeoutConn}[2]{\expiredTimeoutConnSYMB(#1,#2)}

\newcommand{\rtbcProposeSYMB}{\mbox{\texttt{RTBC-propose}}}
\newcommand{\rtbcPropose}[3]{\rtbcProposeSYMB(#1,#2,#3)}

\newcommand{\rtbcDecideSYMB}{\mbox{\texttt{RTBC-decide}}}
\newcommand{\rtbcDecide}[3]{\rtbcDecideSYMB(#1,#2,#3)}

\newcommand{\rtbcInitSYMB}{\mbox{\texttt{RTBC-init}}}
\newcommand{\rtbcInit}[1]{\rtbcInitSYMB(#1)}

\newcommand{\rtbabBcastSYMB}{\mbox{\texttt{RTBAB-broadcast}}}
\newcommand{\rtbabBcastM}{\rtbabBcastSYMB()}
\newcommand{\rtbabBcast}[2]{\rtbabBcastSYMB(#1,#2)}

\newcommand{\rtbabDelivSYMB}{\mbox{\texttt{RTBAB-deliver}}}
\newcommand{\rtbabDelivM}{\rtbabDelivSYMB()}
\newcommand{\rtbabDeliv}[2]{\rtbabDelivSYMB(#1,#2)}

\newcommand{\rtbabInitSYMB}{\mbox{\texttt{RTBAB-init}}}
\newcommand{\rtbabInit}[1]{\rtbabInitSYMB(#1)}

\newcommand{\RTBABunordered}{\META{unordered}}
\newcommand{\RTBABdelivered}{\META{delivered}}
\newcommand{\RTBABnext}{\META{next}}
\newcommand{\RTBABseq}{\META{seq}}
\newcommand{\RTBABbusy}{\META{busy}}
\newcommand{\RTBABinst}{\META{inst}}

\newcommand{\wait}[1]{\texttt{wait}(#1)}
\newcommand{\FALSE}{\texttt{False}}
\newcommand{\TRUE}{\texttt{True}}

\newcommand{\leaderSYMB}{\texttt{leader}}

\newcommand{\leader}[1]{\leaderSYMB(#1)}

\newcommand{\nIter}[2]{\lceil\frac{#1}{#2}\rceil}

\newcommand{\comalgo}[1]{\textit{\textcolor{Maroon}{#1}}}

\newcommand{\deltaRTBC}{\Delta_{\mathtt{C}}}
\newcommand{\deltaRTBRB}{\Delta_{\mathtt{R}}}
\newcommand{\deltaRTBAB}{\Delta_{\mathtt{A}}}
\newcommand{\deltaWait}{\Delta_{\mathtt{W}}}
\newcommand{\deltaB}{\Delta_{\mathtt{B}}}
\newcommand{\deltaP}{\Delta_{\mathtt{P}}}
\newcommand{\deltaD}{\Delta_{\mathtt{D}}}

\newcommand{\fanout}{X}
\newcommand{\total}{N}

\newcommand{\trace}{\tau}

\newcommand{\intitle}[1]{{\smallskip\noindent\textbf{#1}}}

\definecolor{dgreen}{RGB}{0,100,0}
\definecolor{dred}{RGB}{139,0,0}

\definecolor{vrpink}{RGB}{255,0,127}
\definecolor{vrpinkb}{RGB}{255,00,200}
\definecolor{vrblue}{RGB}{30,144,255}
\definecolor{vrolive}{RGB}{85,107,47}
\definecolor{vrroyalblue}{RGB}{65,105,225}
\definecolor{vrlpink}{RGB}{255,192,203}

\newcommand{\newtext}[1]{\textcolor{black}{#1}}

\AtEveryBibitem{\clearname{editor}}

\algnewcommand\algorithmicswitch{\textbf{switch}}
\algnewcommand\algorithmiccase{\textbf{case}}
\algnewcommand\algorithmicassert{\texttt{assert}}
\algnewcommand\Assert[1]{\State \algorithmicassert(#1)}%
\algdef{SE}[SWITCH]{Switch}{EndSwitch}[1]{\algorithmicswitch\ #1\ \algorithmicdo}{\algorithmicend\ \algorithmicswitch}%
\algdef{SE}[CASE]{Case}{EndCase}[1]{\algorithmiccase\ #1}{\algorithmicend\ \algorithmiccase}%
\algtext*{EndSwitch}%
\algtext*{EndCase}%

\ifCLASSINFOpdf
\else
\fi

\begin{document}

\title{PISTIS: An Event-Triggered Real-Time Byzantine-Resilient Protocol Suite\\ \large{(Extended Version)}}

\author{\IEEEauthorblockN{David Kozhaya$^{1}$, J\'er\'emie Decouchant$^{2,*}$, Vincent Rahli$^{3,\dagger}$, and Paulo Esteves-Verissimo$^{4,*}$}
  
  \IEEEauthorblockA{
    $^{1}$ABB Research Switzerland; 
    $^{2}$TU Delft; 
    $^{3}$University of Birmingham;
    $^{4}$KAUST - RC3
      \thanks{$^{*}$Work partly performed while these authors were with the University of Luxembourg.}
      \thanks{$^{\dagger}$Rahli was partially supported by the National Cyber Security Centre (NCSC) project: Aion: Verification of Critical Components’ Timely Behavior in Probabilistic Environments.}
  \vspace*{-20pt}}
}

\maketitle
\thispagestyle{plain}
\pagestyle{plain}
\begin{abstract}
The accelerated digitalisation of society along with technological
evolution have extended the geographical span of cyber-physical
systems.
Two main threats have made the reliable and real-time control of these systems challenging: (i) uncertainty in the communication infrastructure induced by scale,
 and heterogeneity of the environment and devices; and (ii)
targeted attacks maliciously worsening the impact of the
above-mentioned communication uncertainties, disrupting the
correctness of real-time applications.  

This paper addresses those challenges by showing how to
build distributed protocols that provide both real-time with practical
performance, and scalability in the presence of network faults and
attacks, \newtext{in probabilistic synchronous environments}.
We provide a suite of real-time Byzantine protocols, which we
prove correct, starting from a reliable broadcast protocol, called
\emph{PISTIS}, up to atomic broadcast and consensus.
This suite simplifies the construction of powerful distributed and decentralized
monitoring and control applications, including state-machine
replication.
Extensive empirical
\newtext{simulations}
 showcase PISTIS's robustness, latency, and scalability.
For example, PISTIS can withstand message loss (and delay) rates up to
50$\%$
in systems with 49 nodes and provides bounded delivery latencies in the order of a few
milliseconds.
\end{abstract}


\IEEEpeerreviewmaketitle

\newtheorem{theorem}{Theorem}
\newtheorem{example}{Example}
\newtheorem{corollary}{Corollary}
\newtheorem{lemma}{Lemma}
\newtheorem{remark}{Remark}
\newtheorem{definition}{Definition}
\newtheorem{conjecture}{Conjecture}
\newtheorem{observation}{Observation}
\newtheorem{conclusion}{Conclusion}
\newtheorem{assumption}{Assumption}
\newtheorem{property}{Property}
\newtheorem{sketch}{Proof Sketch}

\definecolor{orange}{rgb}{1,0.5,0}

\vspace*{-10pt}\section{Introduction}
\label{intro}

The accelerated digitalisation of society has significantly shifted
the way that physical infrastructures---including large continuous
process plants, manufacturing shop-floors, power grid installations,
and even ecosystems of connected cars---are operated nowadays.
Technological evolution has made it possible to orchestrate a higher
and finer degree of automation, through the proliferation of multiple
sensing, computing, and communication devices that monitor and
control such infrastructures.
These monitoring and control devices are distributed by nature of the
geographical separation of the physical processes they are concerned
with.
The overall systems, i.e., the physical infrastructures with their
monitoring and control apparatus, are generally known as
\textit{cyber-physical systems} (CPS)~\cite{controlarea}.
However, transposing the monitoring and control functionality normally
available in classical, %
real-time (\newtext{i.e., adhering to given time bounds}) and
embedded systems, to the distributed CPS scenarios mentioned above,
is a very challenging task, due to two main reasons.

First, the scale of the systems
 as well as
the heterogeneity of devices (sensors, actuators and gateways), induce
uncertainty in the communication infrastructure interconnecting them,
itself often diverse too, e.g., Bluetooth, Wireless IEEE 802.11, or
Fiber~\cite{endtoendreal-time,controlloss1,eors,D1.1}.
These communication uncertainties become
evident~\cite{controlloss1,eors,D1.1}, namely in the form of
link faults and message delays, which hamper the necessary reliability
and synchronism needed to realize real-time operations, be it when
fetching monitoring data or when pushing decisions to controllers.

Second, security vulnerabilities of many integrated devices, as well as the criticality of the managed physical structures, increase the likelihood of targeted
attacks~\cite{sensorsecurity,sensorsecurity1}. Such attacks can aim to inflict
inconsistencies across system components or to disrupt the timeliness
and correctness of real-time applications. %
The consequences of such attacks can range from loss
of availability to severe physical damage~\cite{blackout}.

This paper addresses the challenges above, %
which render traditional approaches for building real-time
communications, ineffective in wide-scale, uncertain, and vulnerable settings.
We investigate, in particular, how to build large-scale distributed
protocols that can provide real-time communication guarantees and can
tolerate network faults and attacks, \newtext{in probabilistic synchronous environments}.
These protocols simplify the
construction of powerful distributed monitoring and
control applications, including state-machine replication for fault tolerance.  To our knowledge, literature, with the exception of~\cite{flaviu,RT-ByzCast}, has
targeted achieving either real-time guarantees or
Byzantine-resilience with network uncertainties, but not both.

To bridge this gap, we present a protocol suite of real-time Byzantine
protocols, providing several message delivery semantics, from reliable
broadcast (\emph{PISTIS}\footnote{PISTIS was a Greek goddess who
  represented the personified spirit (daimona) of trust, honesty and
  good faith.}), through consensus (\emph{PISTIS-CS}), to atomic
broadcast (\emph{PISTIS-AT}).
PISTIS is capable of: (i) delivering real-time practical performance
\newtext{(i.e., correct nodes provide guarantees within given time bounds)} 
in the presence of aggressive faults and
attacks \newtext{(i.e., one third of the nodes being Byzantine, and
  high message loss rates)}; and (ii) scaling with increasing system size.

The main idea underlying PISTIS is an event-triggered signature based approach
to constantly
monitor the network connectivity among processes. %
Connectivity is measured thanks to the
broadcast messages: processes embed signed monitoring
information within the messages of the broadcast protocol and exclude
themselves from the protocol when they are a threat to
timeliness. Hence, PISTIS does not modularly build on
membership/failure detector oracles (like in traditional distributed
computing) but rather directly incorporates such functionalities within. In
fact, modularity in this sense was proven to be impossible for
algorithms implementing PISTIS-like guarantees~\cite{RT-ByzCast}.  In
order to mask network uncertainties in a scalable manner, PISTIS uses
a temporal and spatial gossip-style message diffusion with fast
signature verification schemes.

We empirically show that PISTIS is robust.
For example PISTIS can tolerate message loss rates of up to
40$\%$, 50$\%$, 60$\%$, and 70$\%$
in systems with 25, 49, 73, and 300 nodes respectively:
PISTIS has a negligible probability of being unavailable under such
losses.  We also show that PISTIS can meet the strict timing constraints of a
large class of typical CPS applications, mainly in \newtext{Supervisory Control
And Data Acquisition (SCADA) and Internet of Things (IoT)}
areas, e.g., (1) fast automatic interactions
($\leq{20}\mbox{ms}$) for systems with up to 200 nodes, (2) power systems and substation
automation
applications ($\leq{100}\mbox{ms}$) for systems with up to 1000 nodes, and (3) slow speed
auto-control functions ($\leq{500}\mbox{ms}$), continuous control
applications ($\leq{1}\mbox{s}$) as well as operator commands of SCADA
applications ($\leq{2}\mbox{s}$) for systems with 1000 nodes or more. Such SCADA and IoT applications could include up to hundreds of devices where reliable and timely communication is required.

By using PISTIS as the baseline real-time Byzantine reliable broadcast
protocol, we prove that (and show how) higher-level real-time
Byzantine resilient abstractions can be modularly implemented, namely,
consensus and atomic broadcast.
Interestingly, we prove that this can be realized with negligible
effort: (1)~we exhibit classes of algorithms which are amenable to
real-time operations by re-using existing synchronous algorithms from
the literature; and (2)~we rely on PISTIS, which addresses and
tolerates the most relevant problems posed by the communication
environment, including the impossibility of modularly handling
membership/failure detection~\cite{RT-ByzCast}.

In short, %
our contributions are:
\begin{itemize}
\item The PISTIS protocol suite\newtext{, which is to the best of our
  knowledge the first generic and modular protocol suite that
  provides message delivery guarantees for protocols ranging from Byzantine
  reliable broadcast to Byzantine atomic broadcast.}
  PISTIS itself is an event-triggered real-time Byzantine
  reliable broadcast algorithm that has higher scalability and faster
  message delivery than conventional time-triggered real-time
  algorithms, in the presence of randomized and unbounded network
  disruptions. Building on top of PISTIS, we present classes of
  algorithms, PISTIS-CS and PISTIS-AT, that implement
  real-time Byzantine consensus and atomic broadcast,  respectively.
\item Correctness proofs of the PISTIS protocol suite.
 We provide the main proof results in this paper (exhaustive proofs are deferred
  \ifx\tAppA\tAppB%
  to Appx.~\ref{appx:correctness-pistis}).
  \else%
  to~\cite[Appx.B]{Kozhaya+Decouchant+Rahli+Verissimo:pistis:long:2019}).
  \fi
\item Extensive empirical
\newtext{simulations} using  Omnet++~\cite{omnet}
 that showcase PISTIS's robustness,
  latency, and scalability. %
\end{itemize}

\intitle{Roadmap.}
The rest of the paper is organized as follows. Sec.~\ref{related work} discusses
related work.  Sec.~\ref{Sysmodel}
details our system model.  Sec.~\ref{sec:rtbrb} recalls the
properties of a real-time Byzantine reliable broadcast, and presents
our algorithm, PISTIS, in details.  Sec.~\ref{sec: system
  transformation} shows and proves how real-time Byzantine atomic broadcast and consensus can be realized on top of PISTIS's guarantees
using classes of existing algorithms.  Sec.~\ref{sec: evaluation}
evaluates the performance and reliability of PISTIS.  Finally, Sec.~\ref{conclusion} concludes the paper. For space limitations, proofs
and additional material are deferred \ifx\tAppA\tAppB%
to Appendices.
\else%
to a dedicated report~\cite{Kozhaya+Decouchant+Rahli+Verissimo:pistis:long:2019}.
\fi

\section{Related Work}
\label{related work}

\newtext{Reliable broadcast is a standard abstraction to ensure that
  the (correct) nodes of a distributed system agree on the delivery
  of messages even in the presence of faulty nodes. Byzantine
  reliable broadcast in particular guarantees that (correct) nodes
  agree even in the presence of arbitrary faults. It is a key building
  block of reliable distributed systems such as Byzantine
  Fault-Tolerant State Machine Replication protocols, which are
  nowadays primarily used in blockchain systems. Pioneered by the work
  of Dolev~\cite{Dolev:focs:1981} and Bracha~\cite{Bracha:1987},
  many protocols have been proposed since then that are intended to
  work in various environments. The focus of our paper is on novel
  Byzantine broadcast primitives and protocols that achieve timeliness
  guarantees.}

This paper has evolved from, and improved over, a research line paved
by~\cite{flaviu,AMP,RT-ByzCast} on timing aspects of reliable
broadcast and Byzantine algorithms. Besides these works, the
literature on broadcast primitives, to the best of our knowledge,
either does not take into account timeliness and maliciousness or
addresses them separately.

Cristian et al.~\cite{flaviu} assumed that all correct processes
remain synchronously connected, regardless of process and network
failures.
This strong network assumption is too optimistic, both in terms of scale and
timing behaviour, which in practice leads to poor
performance (latency of approximately $2.4$ seconds with 25 processes---see Table~\ref{fig:worst-case-latencies} in
Sec.~\ref{sec:PISTIS-overhead} for more details).  Moreover,
Cristian et al.'s system model does not allow processes that malfunction (e.g., by violating timing assumptions) to know that they
are treated as faulty by the model. Our algorithm, in
comparison, provides latencies in the range of few
milliseconds and our model makes processes aware of their untimeliness. %

Verissimo et al.~\cite{AMP} addressed the timeliness problem by
\emph{weak-fail-silence}: despite the capability of the transmission
medium to deliver messages reliably and in real-time, the protocol
should not be agnostic of potential timing or omission faults (even if
sporadic).
The bounded omissions assumption (pre-defined maximum number of
omissions) of~\cite{AMP} could not be taken as is, if
we were to tolerate higher and more uncertain faults (as we consider
in this paper): it could easily lead to system unavailability in
faulty periods. Hence we operate with much higher uncertainty levels
(faults and attacks).

Kozhaya et al.~\cite{RT-ByzCast} devised a Byzantine-resilient
algorithm that provides an upper bound on the delivery latency of messages.
This algorithm is time-triggered and relies on an all-to-all
communication that limits the algorithm's scalability. Our work
improves over~\cite{RT-ByzCast} on several points: (i) we reduce the
delivery latency (few milliseconds as shown in
Fig.~\ref{fig:duration1ms} and Fig.~\ref{fig:duration5ms} compared to a few hundred as shown in~\cite[Fig.~8]{RT-ByzCast}---see also
Table~\ref{fig:worst-case-latencies} for a comparison of worst case
latencies) by adopting an event-triggered approach instead of a
round-based one; (ii) we improve the system's scalability (at least 5
times less bandwidth consumption) by adopting a gossip-based
dissemination instead of an all-to-all communication; and (iii) we
show how real-time broadcast primitives can be modularly used to build
real-time Byzantine-resilient high-level abstractions like consensus
and atomic broadcast.

Guerraoui et al.~\cite{Guerraoui+al:disc:2019} designed a scalable
reliable broadcast abstraction that can also be used in a
probabilistic setting where each of its properties can be violated
with low probability. They achieve a scalable solution by relying on
stochastic samples instead of quorums, where samples can be much
smaller than quorums. As opposed to this work, our goal is to design a
deterministic abstraction where the property are never violated:
\newtext{the real-time Byzantine-resilient reliable broadcast
  primitive discussed in Sec.\ref{sec:rtbrb}} is deterministic because
late processes become passive, and therefore count as being faulty.

In~\cite{Babay+al:dsn:2019,Babay+al:dsn:spire:2018}, the authors
present a Byzantine fault-tolerant SCADA system that relies on the
Prime~\cite{Amir+Coan+Kirsch+Lane:dsn:2008,Amir+Coan+Kirsch+Lane:tdsc:prime:2011}
\newtext{Byzantine Fault Tolerant State Machine
  Replication~\cite{Schneider:csur:1990,Castro+Liskov:osdi:1999}
  (BFT-SMR) protocol} protocol to ensure both safety and latency
guarantees. As opposed to PISTIS, Prime relies on an asynchronous
primary-based BFT-SMR. As opposed to Prime, PISTIS-CS and PISTIS-AT
algorithms are designed modularly from a timely reliable broadcast
primitive; and PISTIS allows slow connections between any processes in
a probabilistic synchronous environment, while Prime relies on the
existence of a ``stable'' timely set of processes.

\section{System and Threat Model}
\label{Sysmodel}

\subsection{System Model}\label{system model}

\textbf{Processes.}
We consider a distributed system consisting of a set
$\mathit{\Pi}=\{p_0, p_1, ..., p_{\total-1}\}$ of $\total>1$ processes.  We
assume that processes are uniquely identifiable
and can use digital signatures to verify the authenticity of messages and enforce their integrity. We denote by $\sigma_i(v)$ the signature of
value $v$ by process $p_i$.  We often write $\sigma_i$, when the
payload is clear from the context.
Processes are synchronous, i.e., the delay for performing a local step
has a fixed known bound (note that this does not apply to faulty
processes---see below).

\medskip
\textbf{Clocks.}
Processes have access to local clocks with a bounded and negligible rate drift to real time. These clocks do not need to be synchronized. 

\medskip
\textbf{Communication.}
Every pair of processes is connected by two logical uni-directional
links, e.g., $p_i$ and $p_j$ are connected by links
$l_{ij}$ and $l_{ji}$. Links can abstract a physical bus or a
dedicated network link.
We assume a \emph{probabilistic synchronous communication model}. %
This means that in any transmission attempt to send a message
over on link $l_{ij}$ (with $i \neq j$) at some time $t$, there is a probability $P_{ij}(t)$ that the message reaches its
destination and within a maximum delay $d$ (known to the processes). $d$ is the upper time bound on non-lossy message delivery and $\epsilon_{1}<1-P_{ij}(t)<\epsilon_{2} \ll 1$ where $\epsilon_{1}$ and $\epsilon_{2}$ are small strictly positive values. Such violations exist in networks, as arguably all communication is
prone to unpredictable disturbances, e.g., bandwidth limitation, bad
channel quality, interference, collisions, and stack
overflows~\cite{eors}.
Our probabilistic synchronous
communication has been shown to be weaker, in some
sense~\cite{neverSayNever:ipdps}, than partial
synchrony~\cite{Dwork:1987}. We further discuss and compare our model
to existing traditional ones
\ifx\tAppA\tAppB%
in Appx.~\ref{app:partial-sync}.
\else%
in~\cite[Appx.~A]{Kozhaya+Decouchant+Rahli+Verissimo:pistis:long:2019}.
\fi
We do not model correlated losses explicitly, as previous works	like~\cite{RT-ByzCast} have shown that such bursts can be mitigated and we leave it up to the applications to define how to deal with late messages (i.e., violating the $d$ delay assumption).%

\subsection{Threat Model}
\label{threat model}

\textbf{Processes.}
We assume that some processes can exhibit arbitrary,
a.k.a. \textit{Byzantine}, behavior. Byzantine nodes can abstract
processes that have been compromised by attackers, or are executing
the algorithm incorrectly, e.g., as a result of some fault (software
or hardware).  A Byzantine process can behave arbitrarily, e.g., it
may crash, fail to send or receive messages, delay messages, send
arbitrary~messages,~etc.

We assume that at most $f=\lfloor\frac{\total-1}{3}\rfloor$ processes can
be Byzantine.  This formula was proved to be an upper bound for
solving many forms of agreement in a variety of models such as in non-synchronous
models~\cite{Dolev:1981:BGS:891722,Fischer+Lynch+Merritt:dc:1986}.

We allow nodes to become \emph{passive} in
case they fail to execute in a timely fashion.
\newtext{As explained in Sec.~\ref{sec:et-rtbyzcast}, passive nodes
  stop executing key events to guarantee timeliness.}
A process that exhibits a Byzantine behavior or that enters the
passive mode (see Sec.~\ref{sec:et-rtbyzcast}) is termed
\textit{faulty}. Otherwise, the process is said to be
\textit{correct}.
Note that passive nodes are considered faulty (at least) during the
time they are passive, but are not counted against the $f$ Byzantine
faults.
Therefore, more than $f$ nodes could be faulty in a system over the
full lifespan of a system (up to $f$ nodes could be Byzantine, and up to $N$ processes could be momentarily passive).\\

\medskip
\vspace*{-10pt}\textbf{Clocks.}
The bounded and negligible rate drift assumption in
Sec.~\ref{system model}
has to hold only on a per protocol execution basis, easily met by
current technology (such as techniques relying on GPS~\cite{clock}
or trusted components~\cite{tcb}). Hence the clock of a
correct
process always behaves as described in Sec.~\ref{system model}.

\medskip
\textbf{Communication.}
We assume that Byzantine processes or network adversaries cannot modify the
content of messages sent on a link connecting correct processes
(implemented by authentication through unforgeable signatures~\cite{unforgsig}).

\section{Real-Time Byzantine Reliable Broadcast}
\label{sec:rtbrb}

We now present our solution to guarantee that correct
  nodes reliably deliver broadcast messages in a timely fashion,
  despite Byzantine nodes, and communication disruptions.
  Sec.~\ref{sec:bcast-abstraction} recalls the properties
  of the real-time Byzantine-resilient reliable broadcast (RTBRB)
  primitive~\cite{RT-ByzCast}. Then,
  Sec.~\ref{sec:pistis-overview} presents a high-level overview of the PISTIS
  event-triggered algorithm, which implements the RTBRB primitive,
  while Sec.~\ref{sec:et-rtbyzcast} provides a detailed presentation
  of PISTIS. Finally, Sec.~\ref{sec:recovery} explains how passive
  nodes can recover and become active again to ensure the liveness of
  the system.

\subsection{Real-time Byzantine Reliable Broadcast Abstraction}
\label{sec:bcast-abstraction}

\begin{definition}[RTBRB]
\label{def:rtbrb}
The real-time Byzantine reliable broadcast (RTBRB) primitive
guarantees the following properties~\cite{RT-ByzCast},
assuming every message is uniquely identified (e.g., using the pair of
a sequence number and a process id---the broadcaster's
id).\footnote{RTBRB's properties are equivalent to the ones of the
  Byzantine reliable broadcast abstraction defined
  in~\cite[Module~3.12,p.117]{Guerraoui:2006:IRD:1137759}, excluding
  \emph{Timeliness}.} In this abstraction, a process broadcasts a message by invoking
$\bcastOpM$.  Similarly, a process delivers a message by
invoking $\delivOpM$.
\begin{itemize}
\item \textbf{\textit{RTBRB-Validity:}} If a correct process $p$
  broadcasts $m$, then some correct process eventually delivers
  $m$.

\item \textbf{\textit{RTBRB-No duplication:}} No correct process
  delivers message $m$ more than once.

\item \textbf{\textit{RTBRB-Integrity:}} If some correct process
  delivers a message $m$ with sender $p_i$ and process $p_i$ is
  correct, then $m$ was previously broadcast by $p_i$.

\item \textbf{\textit{RTBRB-Agreement:}} If some correct process
  delivers $m$, then every correct process eventually
  delivers~$m$.

\item \textbf{\textit{RTBRB-Timeliness:}} There exists a known
  $\deltaRTBRB$ such that if a correct process broadcasts $m$ at
  real-time $t$, no correct process delivers $m$ after
  real~time~$t+\deltaRTBRB$.
\end{itemize}
\end{definition}

It is important to note that the above abstraction does not enforce
ordering on the delivery of messages sent. We elaborate more on that
and how to achieve order in Sec.~\ref{sec: system transformation}.
Note also that in a system consisting of correct and faulty
nodes, these properties ensure that correct nodes deliver broadcast messages
within a bounded delay, while no such guarantee is (and can be) provided about faulty nodes.

\subsection{Overview of PISTIS}
\label{sec:pistis-overview}

This section presents a high-level description of \emph{PISTIS}. %
For simplicity, we assume the total number of processes to
be $\total=3f+1$, in which case a Byzantine quorum has a size of
$2f+1$.  PISTIS guarantees RTBRB properties deterministically despite
the probabilistic lossy network. However, this comes at the price of
PISTIS triggering an entire system fail-safe (shutdown) and a
reinitialization of system state when violating RTBRB-Timeliness is
inevitable. We show later in Sec.~\ref{sec: evaluation} that the
probability of PISTIS causing such system fail-safe (and hence
violating an RTBRB property if fail-safe was not triggered) is
negligible.

\medskip
\textbf{System Awareness.}

Given that broadcasts can be invoked at unknown times, there might
exist a correct process in $\mathit{\Pi}\setminus\{p_i\}$ that is
unaware of $p_i$'s broadcast for an unbounded amount of time after it
was issued, since all links can lose an unbounded number of messages.
The occurrence of such scenarios may hinder the system's ability of
delivering real-time guarantees.
To this end, we require that every process~$p_j$ constantly exchanges
messages with the rest of the system.  This regular message exchange
aims at capturing how well $p_j$ is connected to other processes, and
hence to what extent $p_j$ is up-to-date with what is going on in the
system (and to what extent the system knows about $p_j$'s state).  We
achieve this constant periodic message exchange via a function, which
we call \textit{proof-of-connectivity}.\footnote{Periodic message exchange (heartbeats) has been used to discover the
  network state in many monitoring algorithms~\cite{Aguilera:comm,Dav}}
It requires each process to diffuse heartbeats to the rest of the
system in overlapping rounds: a new round is started every~$d$ time
units, and each round is of a fixed duration~$\mathbb{T}$, where
$d<\mathbb{T}$.
\newtext{(Sec.~\ref{sec: evaluation} shows that $\mathbb{T}=8d$ is a
  reasonably good value, while Sec.~\ref{sec:pistis-props}
  highlights the need for overlapping rounds.)}
A round consists in repeatedly (every $d$ units of time) diffusing a
signed heartbeat message to $\fanout$ other processes. $\fanout$
stands for the number of processes to which a process sends a
message in a communication step. The value of $\fanout$ is fixed at
deployment time (i.e., does not change over the execution of a system)
and can range between $0$ and $\total-1$. It is used to avoid network
congestions by enforcing that processes selectively send their
messages to an arbitrary subset of the system.
\newtext{Each round consists then in repeatedly sending
  $\nIter{\mathbb{T}}{d}$ times a message, each time to $\fanout$
  other nodes.
Note that even though the value of $\fanout$ is fixed, in any given
round the set of $\fanout$ processes to which the message is sent in
every repetition can change such that the union of processes to which
the message is sent in all $\nIter{\mathbb{T}}{d}$ repetitions in that
round covers all processes in the system. This is possible when
$N\leq\fanout\times\nIter{\mathbb{T}}{d}$, which we always guarantee
in practice.}
Heartbeat messages are
uniquely identified by sequence numbers, which are incremented prior
to each round.
On receipt of a heartbeat message, a correct process appends its own
signature to it as well as all other seen signatures relative to that
heartbeat; and sends it to $\fanout$ other processes.
At the end of each round, if a process does not receive at least $2f+1$
signatures (including its own) on its own heartbeat, it enters the passive mode.

\begin{figure}[!t]
  \begin{center}
    \includegraphics[width=0.36\textwidth]{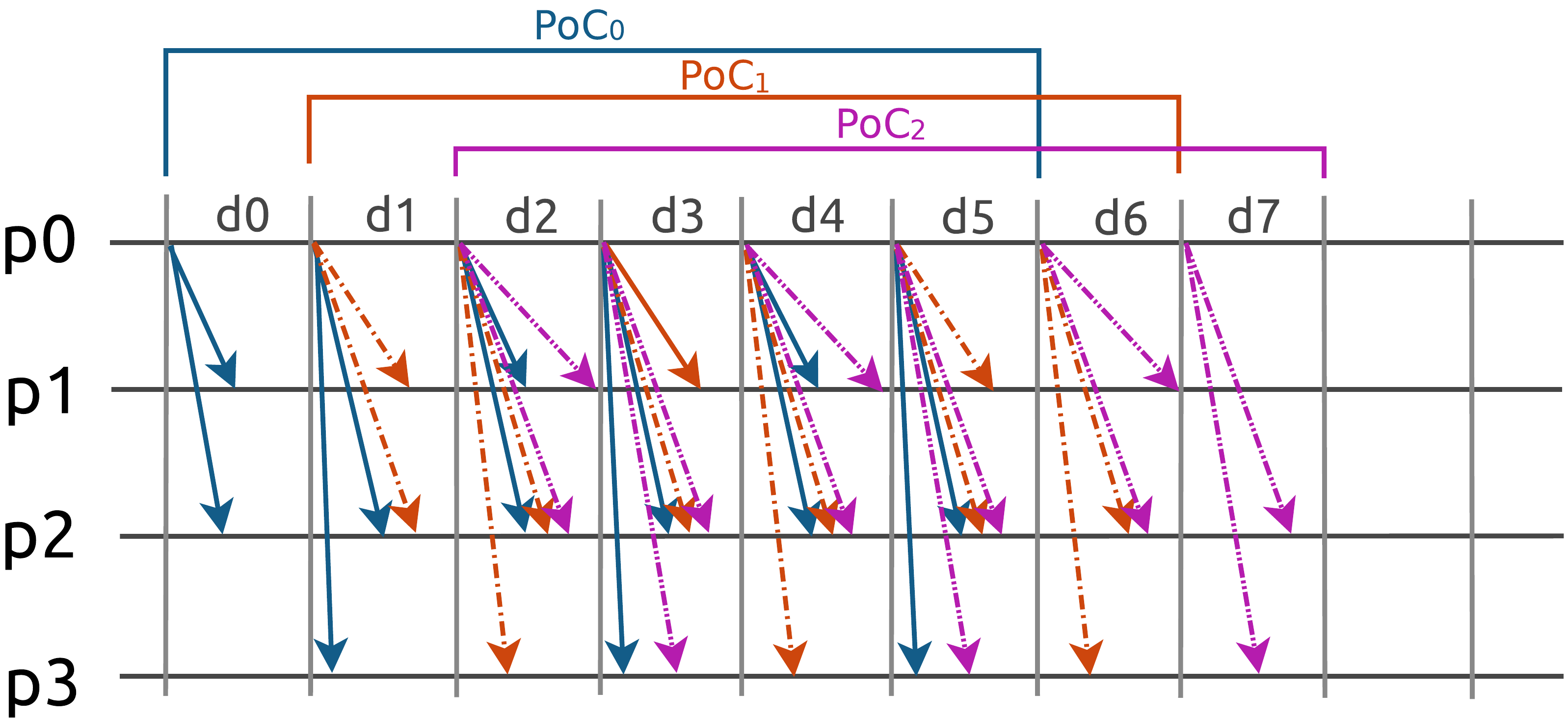}
  \end{center}
  \vspace*{-10pt}
  \caption{Example of a proof-of-connectivity run, where $\fanout=2f+1$,
    and where 2 repetitions allow covering all nodes}
  \label{fig:msd-proof-of-connectivity}
\end{figure}

Fig.~\ref{fig:msd-proof-of-connectivity} provides an example of a run
of the proof-of-connectivity protocol, depicted as a message sequence
diagram, in a system composed of 4 processes. This figure depicts part
of the three first rounds of proof-of-connectivity initiated by $p_0$
(we only show the messages sent by $p_0$ to avoid cluttering the
picture), namely $\META{PoC}_0$ in blue, $\META{PoC}_1$ in orange, and
$\META{PoC}_2$ in purple. In addition, in that case, each proof of
connectivity round is of length $\mathbb{T}=6d$. Therefore, the blue
$\META{PoC}_0$ heartbeats are sent 6 times between $d_0$ and~$d_5$,
the orange $\META{PoC}_1$ heartbeats are sent 6 times between $d_1$
and~$d_6$, and the purple $\META{PoC}_2$ heartbeats are sent 6 times
between $d_2$ and $d_7$. If by the end of $\META{PoC}_0$, $p_0$ has
not received $2f$ replies to its heartbeats, it will become passive.

\medskip
\textbf{Diffusing Broadcasts.} %

PISTIS relies on two types of messages ($\echoSYMB$ and $\delivSYMB$
messages) to ensure that broadcast values are delivered in a timely
fashion.
Processes exchange $\echoSYMB$ messages either to start broadcasting
new values, or in response to received $\echoSYMB$ messages.
$\echoSYMB$ messages help processes gather a valid quorum (a Byzantine
write quorum~\cite{Malkhi+Reiter:stoc:1997} of size $2f+1$) of
signatures on a single value $v$ relative to a broadcast instance. A
broadcast instance is identified by the id of the process broadcasting
$v$ and a sequence number.  $\echoSYMB$ messages help prevent system
inconsistencies when malicious nodes send different values with the
same sequence number (same broadcast instance) to different recipients.
However, additional messages, namely $\delivSYMB$ messages, are needed
to help achieve delivery within a bounded time after the
broadcast.

When a process $p_i$ receives a value $v$ through an $\echoSYMB$
message, it appends its signature to the message as well as all
other signatures it has received relative to $v$; and sends it
to $\fanout$ other processes.
In addition, when $p_i$ receives a value for the first time, it
triggers a local timer of duration~$\mathbb{T}$.
Upon receiving a value signed by more than $2f$
processes, a process delivers that value.  However, a process that does not
receive more than $2f$ signatures on time (i.e., before the timer
expires) enters the passive mode.
In case multiple values are heard relative to a single process and
sequence number (equivocation), then the first heard value is the one
to be echoed.
Note that processes continue executing the proof-of-connectivity
function during the \emph{echo} and \emph{deliver} phases however by
piggybacking heartbeats to echo/deliver messages.

As opposed to $\echoSYMB$ messages that are diffused (i.e.,
re-transmitted temporally and sporadically) for a duration
$\mathbb{T}$, $\delivSYMB$ messages are diffused for $2\mathbb{T}$.
This is needed to ensure that if some correct processes start
diffusing a message between some time $t$ and
$t+\mathbb{T}$, possibly at
different times, then there must be a $\mathbb{T}$-long period of time
where all of them are diffusing the message
\ifx\tAppA\tAppB%
(see Lemma~\ref{lem:intersecting-delivery} in
Appx.~\ref{appx:correctness-pistis}
for more details).
\else%
(see~\cite[Appx.~B, Lem.~4]{Kozhaya+Decouchant+Rahli+Verissimo:pistis:long:2019}
for more details).
\fi
Given a large
enough collection of such processes ($f+1$ correct processes), this
allows other processes to learn about delivered values in a timely
fashion.

\begin{figure}[!t]
  \begin{center}
    \includegraphics[width=0.49\textwidth]{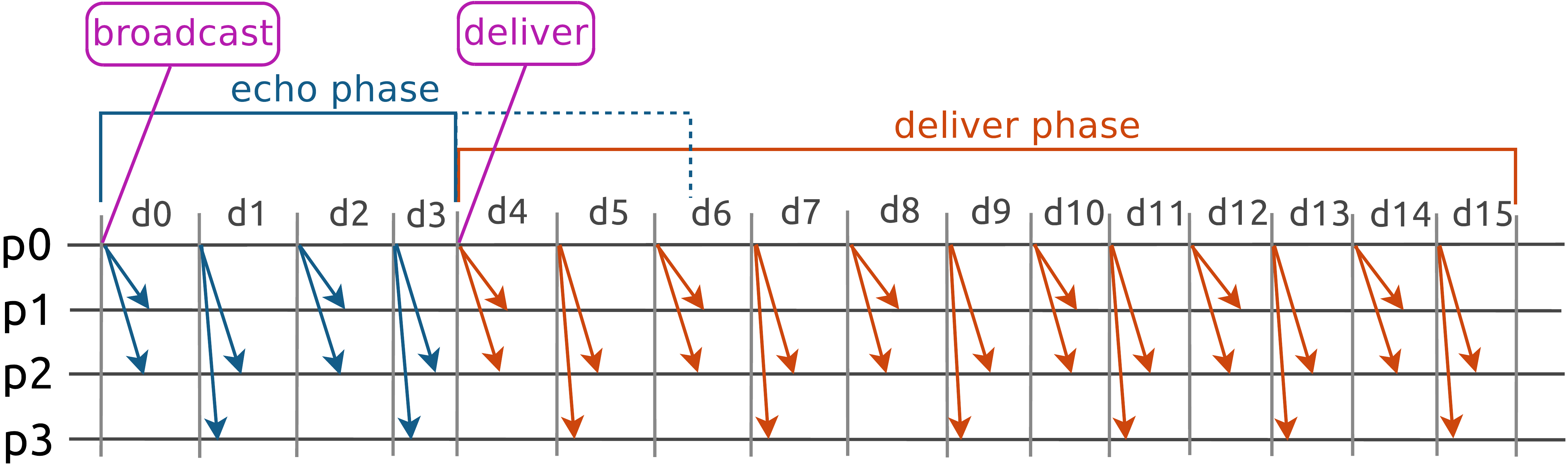}
  \end{center}
  \vspace*{-10pt}
  \caption{Example of a PISTIS run where $\fanout=2f+1$,
    and where 2 repetitions allow covering all nodes}
  \label{fig:msd-pistis}
\end{figure}

Fig.~\ref{fig:msd-pistis} provides an example of a run of PISTIS,
depicted as a message sequence diagram. The system is composed of 4
processes. This figure depicts part of the echo (in blue) and deliver
(in orange) phases of one broadcast initiated by $p_0$ (for the
purpose of this illustration, only the messages sent by $p_0$ are
shown).
The purple ``broadcast'' and ``deliver'' tags indicate the times at which 
$p_0$ initiated its broadcast, and delivered it.
In this example, the echo phase is initially meant to last for a duration of $\mathbb{T}=6d$.
However, it happens here that $p_0$ received $2f$ echo messages for
its broadcast by $3d+k$, where $0<k<d$, which is why $d_3$ is shorter
than the other intervals. Therefore, $p_0$ stops its echo phase and
starts its deliver phase at $3d+k$. As mentioned above, the deliver
phase lasts for $2\mathbb{T}$. If $p_0$ has not received $2f$ deliver
messages in return by the end of that deliver phase, then it becomes
passive.

\floatname{algorithm}{Algorithm}
\begin{algorithm}[!t]
  \caption{\textit{proof-of-connectivity($\mathbb{T}$)} @ process $p_i$}\label{alg:proof-of-life}
  \begin{algorithmic}[1]
    \footnotesize	
    \State $seq = [0]^n$; \comalgo{// stores smallest valid sequence number per process.}
    \State $sq=0$; \comalgo{// local sequence number.}
    \State $\rhb=[\emptyset]^n$; \comalgo{// stores signatures on last $\nIter{\mathbb{T}}{d}$ heartbeats of processes.}
    \State

    \Event{$\initialization\vee\checkConnectivity$}
    \State \textbf{trigger} $\timeOutConn{\META{msg}}{\mathbb{T}}$;
    \State \textbf{Execute} \texttt{h-diffuse$\left(\mkmsg{p_i}{sq},\{\sigma_{i}\}\right)$};\label{alg:line:hdiffuseo}
    \State $\rhb[p_i].add(\mkmsg{p_i}{sq};\{\sigma_{i}\})$; $sleep(d)$; $sq$++;
    \If{$sq-seq[p_i]>\nIter{\mathbb{T}}{d}$} \label{alg:incrementstart}
 $seq[p_i]$++;
    \EndIf\label{alg:incrementend}
    \State  \textbf{trigger} $\checkConnectivity$;
    \EndEvent

    \Event{$\expiredTimeoutConn{\mkmsg{p_i}{sq'}}{\META{timeout}}$}\label{alg:timeoutstart}
    \If{$|\rhb[p_i].getsig(sq')|\leq~2f$}
    \State \comalgo{// gets signatures on message with sequence number $sq'$}
    \State \textbf{Initiate} \textit{passive mode};
    \Else~$\rhb[p_i].remove(sq')$; \comalgo{// remove entry with seq. num. $sq'$}
    \EndIf
    \EndEvent \label{alg:timeoutend}

    \Event {receive $\hb{\mkmsg{p_j}{sq'}}{\Sigma}$} \label{alg:rcvheartbeatstart}
    \If{$(sq'\geq{seq[p_j]})$}
    \State $\rhb[p_j].setsig(sq', \rhb[p_j].getsig(sq') \cup \Sigma \cup \{\sigma_{i}\})$; \label{alg:updatesigs}
    \If{$j\neq i \land sq'\neq seq[p_j]$}
    \State \textbf{Execute} \texttt{h-diffuse$\left(\mkmsg{p_j}{sq'},\rhb[p_j].getsig(sq')\right)$}; \label{alg:diffusehb}
    \EndIf
    \EndIf
    \If{$sq'>(seq[p_j]+\nIter{\mathbb{T}}{d}) \land j\neq{i}$}\label{alg:removehbstart}
    \State $seq[p_j]=sq'-\nIter{\mathbb{T}}{d}$;\label{alg:updlwm}
    \State $\rhb[p_j].remove(sq'')$, $\forall sq''<seq[p_j]$; \label{alg:line:removehba}
    \EndIf \label{alg:removehbend}

    \EndEvent \label{alg:rcvheartbeatend}
    \Function{\texttt{h-diffuse}$\left(\META{msg},\META{\Sigma}\right)$}\label{alg:line:hdiffuseb}%
    \For{$(\texttt{int}\ i=0$; $i\leq\nIter{\mathbb{T}}{d}$; $i$++)}
    \State \textbf{send} $\hb{\META{msg}}{\Sigma}$ to $\fanout$ other processes;
    \State $sleep(d);$
    \EndFor \label{alg:line:hdiffusee}
    \EndFunction
  \end{algorithmic}
\end{algorithm}

\subsection{Detailed Presentation of PISTIS}
\label{sec:et-rtbyzcast}

\floatname{algorithm}{Algorithm}
\begin{algorithm}%
  \caption{PISTIS @ process $p_i$}
  \label{main-alg}
 \begin{algorithmic}[1]
    \footnotesize
    \State \textbf{Execute} \textit{proof-of-connectivity}($\mathbb{T}$);
    \State

    \Event {$\bcastOp{p_i}{\sq}{v}$}\label{alg:bcastA} \label{bcast}
    \State \textbf{Execute} \textit{proof-of-connectivity} in piggyback mode;
    \State \textbf{Initialize} $\recho{p_i}{\sq}{v}=\{\sigma_{i}\}$;\label{alg:send-bcast-init-sign}
    \State \textbf{Execute} $\tDiffuse{\mkmsgv{p_i}{\sq}{v}}{\mathbb{T}}{\echoMode}$;\label{alg:diffuse-broadcast}
    \EndEvent\label{alg:bcastB}

    \Event{receive $\echo{\mkmsgv{p_j}{\sq}{v}}{\Sigma}$}\label{alg:echostart}\label{alg:echo1A}
    \If{$\nexists \recho{p_j}{\sq}{...}$}
    \State \textbf{Initialize} $\recho{p_j}{\sq}{v}=\{\sigma_{i}\}\cup\Sigma$;\label{alg:aggr-new-sign}%
    \State \textbf{Execute} \textit{proof-of-connectivity} in piggyback mode;
    \If {$|\recho{p_j}{\sq}{v}|\leq 2f$}%
    \State \textbf{Execute} $\tDiffuse{\mkmsgv{p_j}{\sq}{v}}{\mathbb{T}}{\echoMode}$;
    \Else~\textbf{Execute} $\deliverMsg{p_j}{\sq}{v}{\recho{p_j}{\sq}{v}}$;\label{alg:line:delmsga}
    \EndIf\label{alg:echo1B}

    \ElsIf{$\exists \recho{p_j}{\sq}{v}$}\label{alg:echo2A}
    \State $\recho{p_j}{\sq}{v}=\recho{p_j}{\sq}{v}\cup\Sigma$;%
    \If {$|\recho{p_j}{\sq}{v}|>2f$ (for the first time)}
    \State \textbf{Execute} $\deliverMsg{p_j}{\sq}{v}{\recho{p_j}{\sq}{v}}$;\label{alg:line:delmsgb}
    \EndIf\label{alg:echo2B}

    \ElsIf{$\exists \recho{p_j}{\sq}{v'\neq{v}}$}\label{alg:echo3A}
    \State \comalgo{// $p_j$ has lied about message with $\sq$}%
    \If {$|\Sigma|> 2f$}
    \State \textbf{remove} $\recho{p_j}{\sq}{v'}$;
    \State $\recho{p_j}{\sq}{v}=\Sigma$;
    \State \textbf{Execute} $\deliverMsg{p_j}{\sq}{v}{\Sigma}$;\label{alg:line:delmsgc}
    \EndIf\label{alg:echo3B}
    \EndIf
    \EndEvent~\label{alg:echoend}

    \Event {receive $\deliv{\mkmsgvs{p_j}{\sq}{v}{\Sigma}}{\Sigma'}$}\label{alg:receive-deliver}
    \If{$\nexists\rdeliv{p_j}{\sq}{v}$}
    \State
    $\recho{p_j}{\sq}{v}=\recho{p_j}{\sq}{v}\cup\Sigma$;
    \State \textbf{Execute}
    $\deliverMsg{p_j}{\sq}{v}{\Sigma}$;\label{alg:line:delmsgd}
    \EndIf
    \State ${\rdeliv{p_j}{\sq}{v}=\rdeliv{p_j}{\sq}{v}\cup\Sigma'}$; \label{alg:receive-deliver1}
    \EndEvent

    \Event {$\expiredTimeout{\META{msg}}{\META{timeout}}{\META{mode}}$}
    \If{$\exists \mathcal{R}_{\META{mode}}(\META{msg})\wedge|\mathcal{R}_{\META{mode}}(\META{msg})| \leq 2f$}
    \Switch{$\META{mode}$}
    \Case{$\echoMode$}
    \If{no lie is discovered on $\META{msg}$}%
    \State \textbf{Initiate} \textit{passive mode};%
    \EndIf
    \EndCase
    \Case{$\delivMode$}
    \State \textbf{Initiate} \textit{passive mode};%
    \EndCase
    \EndSwitch
    \EndIf
    \EndEvent
     \Function{$\tDiffuse{\META{msg}}{\META{timeout}}{\META{mode}}$}
    \State \textbf{trigger} $\timeOut{\META{msg}}{\META{timeout}}{\META{mode}}$;\label{alg:timeout}
    \For{$(\texttt{int}\ i=0$; $i\leq\nIter{\META{timeout}}{d}$; $i$++)}
    \State $\Sigma=\mathcal{R}_{\META{mode}}(\META{msg})$;
    \Switch {$\META{mode}$}
    \Case $\echoMode$
    \State \textbf{send} $\echo{\META{msg}}{\Sigma}$ to $\fanout$ random processes;
    \EndCase
    \Case $\delivMode$
    \State \textbf{send} $\deliv{\META{msg}}{\Sigma}$ to $\fanout$  random processes;
    \EndCase
    \EndSwitch
    \State \texttt{sleep($d$)};
    \EndFor
    \EndFunction
    
    \Function{$\deliverMsgAt{p_i}{p_j}{\sq}{v}{\Sigma}$}\label{alg:line:delmsgfun}
    \If{$\nexists\rdeliv{p_j}{\sq}{v}$}
    \State \textbf{Execute} \textit{proof-of-connectivity} in piggyback mode;
    \State \textbf{trigger }$\delivOp{p_j}{\sq}{v}$;\label{deliver}
    \State \textbf{Initialize} $\rdeliv{p_j}{\sq}{v} = \{\sigma_{i}\}$;
    \State \textbf{Stop} sending any $\echoM$ %
    \EndIf
    \State
    \textbf{Execute} $\tDiffuse{\mkmsgvs{p_j}{\sq}{v}{\Sigma}}{2\mathbb{T}}{\delivMode}$;\label{alg:line:delnewsign}
    \EndFunction
  \end{algorithmic}
\end{algorithm}

We now discuss PISTIS (Algorithm~\ref{main-alg}) in more
details. Note that all functions presented in
  Algorithms~\ref{alg:proof-of-life} and~\ref{main-alg} are
  non-blocking.
PISTIS's proof of correctness can be found
\ifx\tAppA\tAppB%
in Appx.~\ref{appx:correctness-pistis}.
\else%
in~\cite[Appx.B]{Kozhaya+Decouchant+Rahli+Verissimo:pistis:long:2019}.
\fi

\medskip
\textbf{Process states.}
Processes can become passive
 under certain
scenarios by calling ``\textbf{Initiate} \textit{passive mode}''.
\newtext{A passive node stops broadcasting and delivering messages to
  guarantee timeliness but otherwise keeps on replying to messages to
  help other processes.}
Processes that were behaving correctly thus far, are considered faulty
when they initiate a passive mode and can notify the application above
of this fact. Later in this section, we show how processes in the
passive mode can come back to normal operation by calling
``\textbf{Initiate} \textit{active mode}''.

\medskip
\textbf{Ensuring sufficient connectivity.}
In PISTIS every process executes the \textit{proof-of-connectivity}
Algorithm~\ref{alg:proof-of-life}.  Namely, a process~$p_i$ forms a
heartbeat $\hb{\mkmsg{p_i}{sq}}{\{\sigma_{i}\}}$, where $sq$ is
$p_i$'s current heartbeat sequence number and $\sigma_{i}$ is $p_i$'s
signature on $\mkmsg{p_i}{sq}$. Process $p_i$ also stores (in array
$\rhb$) for every process (including itself) all signatures it
receives on heartbeats with a \textit{valid sequence
  number}.
\newtext{A valid heartbeat sequence number for some process $p_j$ is a
  sequence number $\geq{seq[p_j]}$. Heartbeats with lower sequence
  numbers are simply ignored. To avoid receiving heartbeats from older
  rounds, we update $seq[p_j]$ every time a heartbeat with a sequence
  number over $seq[p_j]+\nIter{\mathbb{T}}{d}$ is receiver
  (lines~\ref{alg:removehbstart}--\ref{alg:updlwm}).}
After forming its heartbeat, $p_i$ sets a timeout of duration
$\mathbb{T}$, and sends this heartbeat to $\fanout>f$ random processes
$\nIter{\mathbb{T}}{d}$ times
(lines~\ref{alg:line:hdiffuseb}--\ref{alg:line:hdiffusee}).  Process
$p_i$ increments its heartbeat sequence number and repeats this whole
procedure every $d<\mathbb{T}$. Upon incrementing its heartbeat
sequence number, $p_i$ updates its own valid heartbeat sequence
numbers (lines~\ref{alg:incrementstart}--\ref{alg:incrementend}).

A process $p_i$ receiving $\hb{\mkmsg{p_j}{sq'}}{\Sigma}$ ignores this
heartbeat if $sq'$ is smaller than the smallest valid heartbeat
sequence number known for $p_j$. Otherwise, $p_i$ updates $p_j$'s
valid heartbeat sequence numbers
(lines~\ref{alg:removehbstart}--\ref{alg:removehbend}) and the list of
all seen signatures on these valid heartbeats
(line~\ref{alg:updatesigs}). Then, $p_i$ diffuses the heartbeat with
the updated list of seen signatures to $\fanout$ random processes
(line~\ref{alg:diffusehb}).

When a timer expires, $p_i$ checks $\rhb[p_i]$ for the number of
accumulated signatures on its corresponding heartbeat. If that number
is $\leq 2f$, $p_i$ enters the passive mode; otherwise it removes the
corresponding entry from $\rhb[p_i]$
(lines~\ref{alg:timeoutstart}--\ref{alg:timeoutend}).

\medskip
\textbf{Broadcasting a message.}
A process $p_i$ that wishes to broadcast a value $v$, calls
$\bcastOp{p_i}{\sq}{v}$ from Algorithm~\ref{main-alg}
(lines~\ref{alg:bcastA}--\ref{alg:bcastB}), where $\sq$ is a sequence
number that uniquely identifies this broadcast instance.
Given such an event, $p_i$ produces a signature $\sigma_{i}$ for the
payload $\mkmsgv{p_i}{\sq}{v}$.  It then triggers a timeout of duration
$\mathbb{T}$ and sends an
$\echo{\mkmsgv{p_i}{\sq}{v}}{\{\sigma_{i}\}}$ message
$\nIter{\mathbb{T}}{d}$ times to $\fanout$ other random
processes. \textit{Proof-of-connectivity} information from $p_i$ is
now piggybacked on these messages, as on all other $\echoSYMB$ and
$\delivSYMB$ messages.

\medskip
\textbf{Sending and Receiving Echoes.}
When $p_i$ receives an $\echo{\mkmsgv{p_j}{\sq}{v}}{\Sigma}$, $p_i$
reacts differently depending on whether it is not already echoing for
this instance (lines~\ref{alg:echo1A}--\ref{alg:echo1B}), already
echoing $v$ (lines~\ref{alg:echo2A}--\ref{alg:echo2B}), or already
echoing a different value (lines~\ref{alg:echo3A}--\ref{alg:echo3B}).
In all three cases, $p_i$ starts delivering a message (and stops
sending echoes) as soon as at least $2f+1$ distinct signatures have
been collected for that message.

\medskip
\textbf{Sending and Receiving Deliver Messages.}
When $p_i$ receives $\deliv{\mkmsgvs{p_j}{\sq}{v}{\Sigma}}{\Sigma'}$
for the first time
(lines~\ref{alg:line:delmsgfun}--\ref{alg:line:delnewsign}),
it delivers $\mkmsgvs{p_j}{\sq}{v}{\Sigma}$, and sends
$\deliv{\mkmsgvs{p_j}{\sq}{v}{\Sigma}}{\rdeliv{p_j}{\sq}{v}}$ using
$\tDiffuseF$.
In case that \emph{deliver} message is not the first one received
(lines~\ref{alg:receive-deliver}--\ref{alg:receive-deliver1}), $p_i$
aggregates all seen signatures for $\mkmsgv{p_j}{\sq}{v}$ in
$\rdeliv{p_j}{\sq}{v}$ (all functions that use $\rdeliv{p_j}{\sq}{v}$
now use the new updated value).

\medskip
\textbf{Process Passive Mode.}
When a timeout set by process $p_i$ with parameters
$(\META{msg},\META{timeout},\META{mode})$ expires, $p_i$ enters the passive mode if the set $\mathcal{R}_{\META{mode}}$ has less than $2f+1$
distinct signatures, for $\META{mode}=\delivMode$.  For
$\META{mode}=\echoMode$, $p_i$ enters passive mode if in addition to
$\mathcal{R}_{\META{mode}}$ not having $2f+1$ signatures, $p_i$ did
not discover a lie for that broadcast instance.
\begin{remark}
	\label{rem:valid-message}
	Any message of the form
	$\echo{\mkmsgv{p_j}{\sq}{v}}{\Sigma_1}$ or
	$\deliv{\mkmsgvs{p_j}{\sq}{v}{\Sigma_2}}{\Sigma_3}$ is termed
	\textit{invalid} if: (1)~$\Sigma_1$ contains an incorrect signature,
	and similarly for $\Sigma_2$ and $\Sigma_3$; or (2)~$\Sigma_1$
	does not contain a signature from $p_j$, and similarly for
	$\Sigma_2$; or (3)~$\Sigma_2$ has less than $2f+1$ signatures.
	Invalid messages are simply discarded.
\end{remark}

\begin{remark}
We assume that processes sign payloads of the form
$(p_i,\sq,v,\mathsf{E})$ for echo messages and of the form
$(p_i,\sq,v,\mathsf{D})$ for deliver messages. We use the $\mathsf{E}$
and $\mathsf{D}$ tags to distinguish echo and deliver payloads,
thereby ensuring that an attacker cannot use echo signatures as
deliver signatures. Note that echo signatures are sent as part of
deliver messages as a proof that a quorum of processes echoed a
certain value.
\end{remark}

\subsection{PISTIS' properties}
\label{sec:pistis-props}

As mentioned at the beginning of this section, PISTIS is correct in
the sense that it satisfies all five properties of the RTBRB primitive
presented in Sec.~\ref{sec:bcast-abstraction}:
\begin{theorem}[Correctness of PISTIS]\label{theorem:correctness-pistis}
  Under the model presented in Sec.~\ref{Sysmodel}, the PISTIS
  algorithm presented in Fig.~\ref{main-alg} implements the RTBRB
  primitive.
\end{theorem}

A proof of this theorem can be found
\ifx\tAppA\tAppB%
in Appx.~\ref{appx:correctness-pistis}.
\else%
in~\cite[Appx.B]{Kozhaya+Decouchant+Rahli+Verissimo:pistis:long:2019}.
\fi
Let us point out here that the $\deltaRTBRB$ bound of the
RTBRB-Timeliness property turns out to be $3\mathbb{T}$.

\newtext{Let us also highlight the crux of this proof here.
\begin{center}
\includegraphics[width=0.35\textwidth]{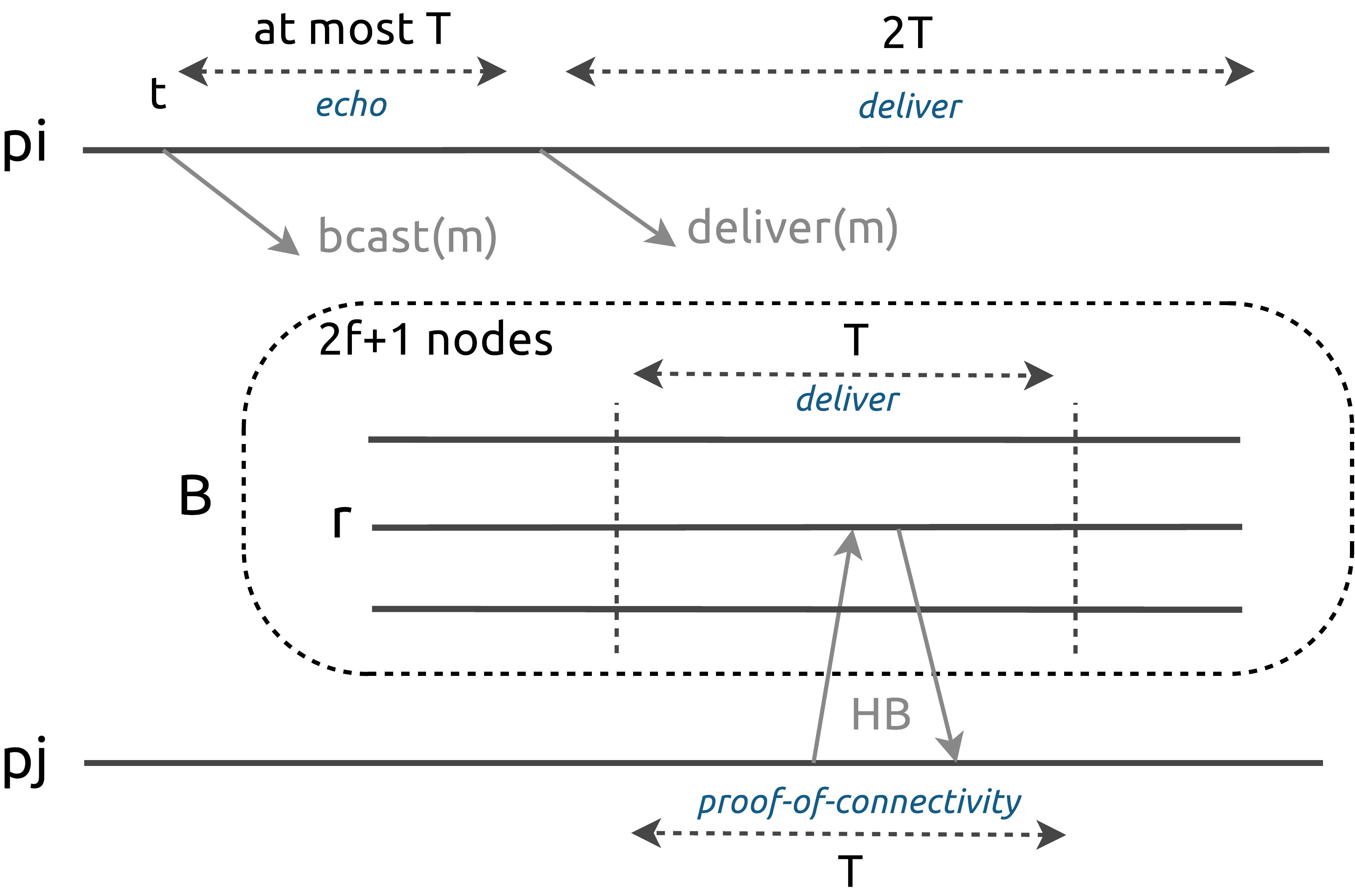}
\end{center}
As illustrated above, a correct node $p_i$ that
broadcasts a message $m$ a time $t$ is guaranteed to start delivering
$m$ by $t_d=t+\duration$. In addition thanks to the $2\duration$
delivery period, we are also guaranteed that a collection, called $B$,
of $2f+1$ nodes, will only deliver $m$ for a $\duration$-long period that
starts before $t_d+\duration$. PISTIS's proof-of-connectivity (PoC)
mechanism then ensures that any other correct node $p_j$ will execute
a PoC round during which a correct node $r\in{B}$ delivers $m$ to
$p_j$, piggybacked to a heartbeat, thereby guaranteeing that $p_j$
delivers $m$ timely.}

\newtext{In particular, overlapping PoC rounds allow for all correct
  nodes to have a PoC round that coincide with that $\duration$-long
  period (called $D$ here), during which the correct nodes in $B$
  deliver $m$, thereby allowing all correct nodes to deliver $m$. If
  PoC rounds were consecutive and not overlapping, a correct
  node could miss the deliver message (piggybacked with PoC messages)
  sent during $D$ if it were to receive PoC messages for a round
  (i.e., sequence number) $s$ sent before $D$, and for round $s+1$
  sent after $D$, thereby staying active while not~delivering.}

\subsection{Byzantine-Resilient Recovery}
\label{sec:recovery}

If process $p_i$ detects that it is executing under bad network
conditions, it enters the passive mode and signals the upper
application. As a result, $p_i$ stops broadcasting and delivering
broadcast messages (by not executing line~\ref{bcast} and
line~\ref{deliver}) to avoid violating RTBRB-Timelines. However, $p_i$
continues participating in the dissemination of the broadcast and
proof-of-connectivity messages to avoid having too many nodes not collecting enough messages and hence becoming passive.

Once the network conditions are acceptable again, $p_i$ can recover
and resume delivering broadcast messages.
More precisely, a process $p_i$ that enters passive mode at time $t$
can operate normally again
if the interval $[t,t+\deltaRTBRB]$ is free of any passive mode
initiations.
 This $\deltaRTBRB$ duration ensures that the messages
delivered by a recovered process $p_i$ do not violate any RTBRB
properties.
After a delay $\deltaRTBRB$, nodes will resume their full
participation in the protocol, and either deliver messages or stay on
hold.

Note that in case of multiple broadcast instances, passive nodes that become active again should learn the latest sequence number of broadcasts for other nodes. Otherwise Byzantine nodes can exploit this to hinder the liveness of the system.

\begin{remark}
Given that processes can now shift between passive and active modes,
we specify our notion of correct processes as follows. A system run is
modeled by a trace of events happening during that run. An event has a
timestamp and a node associated with it. Moreover, an event can either
be a correct event or a Byzantine event.
Given an algorithm~$A$, a process $p$ is deemed correct w.r.t.\ $A$
and a trace $\trace$, if: (1)~it follows its specification from $e_1$,
the first correct $A$-related event (i.e., an event of algorithm~$A$)
happening in $\trace$, to $e_2$, the last correct $A$-related event
happening in $\trace$; (2)~$p$'s events between $e_1$ and $e_2$ must
all be correct; (3)~$p$ must also have followed its specification
since it last started; and (4)~$p$ must never have lost its keys (so
that no other node can impersonate $p$ when $p$ follows its
specification).
The results presented below also hold for this definition of
correctness, because correct processes are required to be active
through the entire broadcast instance.
\end{remark}

This recovery mechanism improves the overall resilience of the system.
Indeed, having all processes in passive mode can occur if $2f+1$ nodes
are passive, which is now harder to achieve if nodes can recover
sufficiently fast enough.

\section{Beyond a Reliable Broadcast}
\label{sec: system transformation}

Unlike liveness in asynchronous reliable broadcast, the
RTBRB-Timeliness property (a safety property) introduces a scent of
physical ordering.  This ordering is due to the fact that timeliness
stipulates, for each execution, a termination event to occur ``at or
before'' some $\deltaRTBRB$ on the time-line.  This said, the reader
may wonder to
what extent does the real-time Byzantine-resilient reliable broadcast
(of Sec.~\ref{sec:bcast-abstraction}) help in establishing total
order?

The answer to this question lies in examining what happens to multiple
broadcasts issued by the same or by different nodes.  When multiple
broadcasts interleave, e.g., when they are issued within a period
shorter than $\deltaRTBRB$ (the upper time bound on delivering a message),
messages might be delivered to different processes in different
orders. The timeliness property of the real-time Byzantine-resilient
reliable broadcast only ensures that a message $m$ that is broadcast
at time $t$ is delivered at any time in $[t,t+\deltaRTBRB]$. Thus, to
ensure total order on all system events, e.g., for implementing
\textit{State Machine Replication}, additional abstractions need to be
built on top of the real-time Byzantine-resilient reliable broadcast
primitive that we have developed so far.

In this section, we investigate how to modularly obtain such an order
on system events while still preserving real-time and
Byzantine-resilience.
We define two build blocks that build on top of RTBRB, namely the RTBC
real-time Byzantine consensus abstraction (Def.~\ref{def:rtbc})---a
fundamental building block for state machine replication, atomic
broadcast and leader
election~\cite{Guerraoui:2006:IRD:1137759}; and the RTBAB
real-time atomic broadcast abstraction (Def.~\ref{def:rtbab})---to
establish total order on system events.
We then provide characterizations of classes of algorithms that
implement these abstractions: Thm.~\ref{theorem:RTBC-implemented}
provides a characterization of the PISTIS-CS class of algorithms that
implement RTBC, while Thm.~\ref{theorem:RTBAB} provides a
characterization of the PISTIS-AT class of algorithms that implement
RTBAB.
Finally, we provided examples of algorithms that belong to these
classes (see Examples~\ref{ex:rtbc} and~\ref{ex:rtbab}).

We start with the following assumption
that constrains the ways processes
can communicate.

\begin{assumption}
  \label{assump:network-access}
  Correct processes access the network only via the RTBRB primitive,
  namely using the two operations: $\bcastOpM$ and $\delivOpM$.
\end{assumption}

From Assumption~\ref{assump:network-access}, a correct process $p_i$ that receives a message from an
operation other than $\delivOpM$ simply ignores that
message by dropping it.

\subsection{Real-Time Byzantine Consensus}
\label{sec: RT-consensus}

Roughly speaking,
solving the \textit{Byzantine consensus} problem requires the agreement of distributed processes on a given value, even though some of the processes may fail
arbitrarily. Byzantine consensus was first identified by Pease et
al.~\cite{Pease:1980}, and formalized as the \textit{interactive
  consistency} problem.  An algorithm achieves interactive consistency
if it allows the non-faulty processes to come to a consistent view of
the initial values of all the processes, including the faulty ones.
Once interactive consistency has been reached, the non-faulty
processes can reach consensus by applying a deterministic averaging or
filtering function on the values of their view. We
apply the following assumption to reach consensus.
\begin{assumption}
  \label{assump:consensus-function}
  Once interactive consistency terminates, every correct process scans
  the obtained vector and decides on the value that appears at least
  $2f+1$ times. If no such value exists, then the process decides
  $\bot$, a distinguished element that indicates that no value has
  been decided.
\end{assumption}

\begin{definition}[RTBC]
\label{def:rtbc}
The real-time Byzantine consensus (RTBC) abstraction is expressed by
the following properties:\footnote{The properties of RTBC are the same
  as the ones of the traditional (strong) Byzantine consensus defined
  in~\cite{Dwork:1987} (see
  also~\cite[Module~5.11,p.246]{Guerraoui:2006:IRD:1137759}),
  excluding the \emph{Timeliness} property.}
\begin{itemize}
\item \textit{\textbf{RTBC-Validity:}} If all correct processes
  propose the same value $v$, then any correct process that decides,
  decides~$v$.  Otherwise, a correct process may only decide a value
  that was proposed by some correct process or
  $\bot$.
\item \textit{\textbf{RTBC-Agreement:}} No two correct processes
  decide differently.
\item \textit{\textbf{RTBC-Termination:}} Correct processes
  eventually decide.
\item \textit{\textbf{RTBC-Timeliness:}} If a correct process $p_i$
  proposes a value to consensus at time $t$, then no correct process
  decides after $t+\deltaRTBC$.
\end{itemize}
\end{definition}

In RTBC a process $p_i$ can propose a value $v$ to consensus by
invoking $\rtbcPropose{p_i}{\META{inst}}{v}$, where $\META{inst}$ is a
sequence number that uniquely identifies a RTBC instance.  Similarly,
a process $p_i$ decides on a value $v$ by invoking
$\rtbcDecide{p_i}{\META{inst}}{v}$.
In addition $\rtbcInit{\META{inst}}$ instantiate a new instance of
RTBC with id $\META{inst}$, i.e., for sequence number $\META{inst}$.

\begin{definition}
	\label{def:bounded}
An algorithm is said to be \emph{bounded} if it only uses a known
	bounded number of communication rounds.
\end{definition}

\begin{theorem}[Characterization of the PISTIS-CS class]
  \label{theorem:RTBC-implemented}
  Let PISTIS-CS be the class of \emph{bounded}
  (Def.~\ref{def:bounded}) algorithms that implements interactive
  consistency under Assumptions~\ref{assump:network-access}
  and~\ref{assump:consensus-function}. Then, PISTIS-CS algorithms also
  implement RTBC in our model (described in Sec.~\ref{Sysmodel}).
\end{theorem}

\ifx\tAppA\tAppB%
See Appx.~\ref{appx:proofRTBC}
\else%
See~\cite[Appx.~C]{Kozhaya+Decouchant+Rahli+Verissimo:pistis:long:2019}
\fi
for a proof of this result.

\begin{example}[Examples of PISTIS-CS algorithms]
\label{ex:rtbc}
Because the interactive consistency problem has been solved using
different algorithms that satisfy Def.~\ref{def:bounded}, our result
applies to various existing algorithms, such
as~\cite{Pease:1980,authenticatedByzagree,boundsbyzagree,byzantine-generals-problem}.
\end{example}

\subsection{Real-Time Byzantine-Resilient Atomic Broadcast}
\label{sec:atomic-bcast}

\begin{definition}[RTBAB]
\label{def:rtbab}
A real-time Byzantine-resilient atomic broadcast (RTBAB) has the same
properties as RTBRB (with a different timeliness bound) plus an
additional ordering property (therefore, we only present the
properties that differ from RTBRB's):
\begin{itemize}
\item \textbf{\textit{RTBAB-Timeliness:}} There exists a known
  $\deltaRTBAB$ such that if a correct process broadcasts $m$ at time
  $t$, no correct process delivers $m$ after
  real~time~$t+\deltaRTBAB$.
\item \textbf{\textit{RTBAB-Total order:}} Let $m_1$ and $m_2$ be any
  two messages and suppose that $p_i$ and $p_j$ are any two correct
  processes that deliver $m_1$ and $m_2$. If $p_i$ delivers $m_1$
  before $m_2$, then $p_j$ delivers $m_1$ before~$m_2$.
\end{itemize}
\end{definition}

We now define the class of algorithms (called $\roundbased$), through
the properties listed below, that modularly implement RTBAB
properties. $\roundbased$ algorithms make use of a single RTBRB
instance and multiple instances of RTBC. We first constrain a
$\roundbased$ algorithm to start an RTBRB instance within a bounded
amount of time for any broadcast call.

\begin{property}\
	\label{prop:rtbab-bcast-rtbrb}
	If a correct process $p_i$ RTBAB-broadcasts a message $m$ at time
	$t$, then it also RTBRB-broadcasts $m$ by time $t+\deltaB$, for
	some bounded $\deltaB$.
\end{property}

We then require a $\roundbased$ algorithm to start (or end in
  case this has already been done before) an RTBC instance, within a
  bounded amount of time, every time the RTBRB instance delivers.

\begin{property}
	\label{prop:rtbrb-deliv-rtbc-prop}
	If a correct process RTBRB-delivers a message $m$ at time $t$, such
	that $m$'s broadcaster is also correct, then it either RTBC-proposes
	or RTBC-decides $m$ by $t+\deltaP$, for some bounded
	$\deltaP$.
\end{property}

In addition, the next property constrains the values that can be
proposed at each RTBC instance, namely that at most one non-$\bot$
value can be proposed at each instance.

\begin{property}
	\label{prop:rtbc-prop-val-or-bot}
	Given an RTBC instance $\RTBABinst$, there exists a value $v$, such
	that each correct process either RTBC-propose $v$ or $\bot$ at
	$\RTBABinst$.
\end{property}

Next, we require a $\roundbased$ algorithm to deliver a RTBC-decided
value within a bounded amount of time
(Property~\ref{prop:rtbc-deliv-rtbab}) and to ensure that
non-RTBC-decided values are re-proposed in later RTBC rounds
(Property~\ref{prop:rtbc-repropose}).

\begin{property}
  \label{prop:rtbc-deliv-rtbab}
  If a correct process RTBC-decides a message $m$ at time $t$, then it
  also RTBAB-delivers $m$ by time $t+\deltaD$, for some bounded
  $\deltaD$.
\end{property}

\begin{property}
	\label{prop:rtbc-repropose}
	A correct process $p_i$ that proposes a value $v$ at a given time
	$t$, using a given RTBC instance $\RTBABinst$, and such that this
	instance does not decide $v$, also RTBC-propose $v$ at some instance
	$\RTBABinst+k$, where $0<k$.
	Moreover, $p_i$ RTBC-proposes $v$ at the smallest instance between
	$\RTBABinst+1$ and $\RTBABinst+k$ where $m$ is proposed by some
	process.
\end{property}

Finally, we require that nodes participate in all successive RTBC instances in a monotonic fashion.

\begin{property}
	\label{prop:rtbc-unique}
	Correct processes RTBC-propose exactly one value per RTBC instance;
	propose values in all RTBC instances (i.e., for all instances
	$\RTBABinst\in\mathbb{N}$); in increasing order w.r.t.\ the instance
	numbers of the RTBC instances (i.e., if $p_i$ proposes values at
	times $t_1$ and $t_2$ using the RTBC instances $\RTBABinst_1$ and
	$\RTBABinst_2$, respectively, and $t_1<t_2$, then
	$\RTBABinst_1<\RTBABinst_2$); and not in parallel (i.e., if $p_i$
	proposes a value at time $t$ using an RTBC instance $\RTBABinst$,
	and that this RTBC instance has not decided by time $t'>t$, then
	$p_i$ does not propose any other value between $t$ and $t'$).
\end{property}

\begin{definition}
  \label{def:round-based}
  Let $\roundbased$ be the class of \emph{round-based}
  algorithms that satisfy the
  properties~\ref{prop:rtbab-bcast-rtbrb},
  \ref{prop:rtbrb-deliv-rtbc-prop},
  \ref{prop:rtbc-prop-val-or-bot},
  \ref{prop:rtbc-deliv-rtbab},
  \ref{prop:rtbc-repropose},
  and~\ref{prop:rtbc-unique}.
\end{definition}

\begin{theorem}[Characterization of the PISTIS-AT class]
  \label{theorem:RTBAB}
  Let PISTIS-AT be the class of $\roundbased$ algorithms that
  implement the traditional Byzantine total-order broadcast under
  Assumption~\ref{assump:network-access}. Then, PISTIS-AT algorithms
  also implement RTBAB in our system (described in
  Sec.~\ref{Sysmodel}).
\end{theorem}

To prove Theorem~\ref{theorem:RTBAB}, it is sufficient to prove that a
RTBAB-broadcasted value $m$ is always RTBAB-delivered within a bounded
amount of time. Because of the round-based property, $m$ must be
RTBRB-proposed and RTBRB-decided within a bounded amount of time.
Consequently there is (within a bounded amount of time) an RTBC
instance where ``enough'' correct nodes RTBC-propose $m$, so that $m$ gets
RTBC-decided upon and RTBAB-delivered within a bounded amount of time.
The proof of Theorem~\ref{theorem:RTBAB} is detailed
\ifx\tAppA\tAppB%
in Appx.~\ref{sec:pistis-to}.
\else%
in~\cite[Appx.~D]{Kozhaya+Decouchant+Rahli+Verissimo:pistis:long:2019}.
\fi

\newtext{We have introduced bounds for each of the
  operations executing in bounded time, namely $\deltaRTBRB$
  (Def.~\ref{def:rtbrb}), $\deltaRTBC$ (Def.~\ref{def:rtbc}),
  $\deltaWait$ (Alg.~\ref{alg:RTBAB}), $\deltaB$
  (Prop.~\ref{prop:rtbab-bcast-rtbrb}), $\deltaP$
  (Prop.~\ref{prop:rtbrb-deliv-rtbc-prop}), $\deltaD$
  (Prop.~\ref{prop:rtbc-deliv-rtbab}), and $\deltaRTBAB$
  (Def.~\ref{def:rtbab}). Those bounds are not assumed to be related
  to each other. However, the bound for $\deltaRTBAB$ we exhibit in
  Theorem~\ref{theorem:RTBAB}'s proof is a combination of all the
  other bounds discussed above.}

\begin{example}[Example of a PISTIS-AT algorithm]
\label{ex:rtbab}
Finally, algorithm~\ref{alg:RTBAB} provides an example of a PISTIS-AT
algorithm that implements RTBAB modularly,
which we adapted from~\cite[Alg.6.2,p.290]{Guerraoui:2006:IRD:1137759} to guarantee
timeliness.
\end{example}

\floatname{algorithm}{Algorithm}
\begin{algorithm}[!t]
	\caption{Example of a \textit{PISTIS-AT} algorithm @process $p_i$}
	\label{alg:RTBAB}
	\begin{algorithmic}[1]
		\footnotesize 

		\Event{$\rtbabInit{\mbox{rtbab}}$}
		\State $\RTBABunordered = []^n$; $\RTBABnext = [0]^n$; $\RTBABseq = 0$;	
		\State $\RTBABdelivered=\emptyset$; $\RTBABbusy = \FALSE$; $\RTBABinst = 0$;
		\EndEvent

		\Event{$\rtbabBcast{p_i}{m}$}\label{pistis-to-event}
		\State \textbf{trigger} $\bcastOp{p_i}{\RTBABseq}{m}$;\label{pistis-to-rtbrb-bcast}
                \State $\RTBABseq$++;
		\EndEvent

		\Event{$\delivOp{p_j}{\META{num}}{m}$}\label{alg:rtbrbDelA}
		\If{$\META{num}=next[p_j]$}
		\State $\RTBABnext[p_j] = \RTBABnext[p_j]+1$;
		\If{$m\notin\RTBABdelivered$}
		\State $\RTBABunordered[p_j] = \RTBABunordered[p_j]\textbf{.append}(\pair{p_j}{m})$;\label{alg:append-new-message}
		\EndIf
		\Else~$\{\wait{\deltaWait};~\text{\textbf{trigger}}~\delivOp{p_j}{\META{num}}{m};\}$\label{alg:wait-deliv}
		\EndIf
		\EndEvent\label{alg:rtbrbDelB}

		\Event{$\exists p_j:\RTBABunordered[p_j]\neq[] \wedge \RTBABbusy=\FALSE$}\label{alg:exists-unordered}
		\State $\RTBABbusy=\TRUE$;
                \State \textbf{trigger} $\rtbcInit{\RTBABinst}$;
		\State \comalgo{// initiate a new real-time Byzantine consensus instance}
		\If{$\RTBABunordered[leader(\RTBABinst)] \neq []$}
		\State $m=\RTBABunordered[\leader{\RTBABinst}]\textbf{.head}()$;
		\Else~$\{m=\bot$;\}
		\EndIf

		\State \textbf{trigger} $\rtbcPropose{p_i}{\META{inst}}{m}$;
		\EndEvent

		\Event{$\rtbcDecide{p_i}{\RTBABinst'}{\META{decided}}$}\label{alg:rtbc-decide}
		\If{$\RTBABinst'=\RTBABinst$}
		\If{$\META{decided}\notin\RTBABdelivered \wedge \META{decided}\neq\bot$}\label{alg:check-if-delivered}
		\State $\RTBABdelivered = \RTBABdelivered \cup \{\META{decided}\}$;\label{alg:add-to-delivered}
		\State \textbf{trigger} $\rtbabDeliv{\leader{\RTBABinst}}{\META{decided}}$;
		\EndIf
		\State $\RTBABunordered[\leader{\RTBABinst}]\textbf{.remove}(\META{decided})$;
		\State $\RTBABinst$++; $\RTBABbusy=\FALSE$;
		\Else~$\{\wait{\deltaWait};~\mbox{\textbf{trigger}}\ \rtbcDecide{p_i}{\RTBABinst'}{\META{decided}}; \}$\label{alg:wait-decide}
		\EndIf
		\EndEvent
		\Function {$\leader{\META{instance}}$}~$\{ \text{\textbf{return}} (\META{instance}\mod{n});\}$ \EndFunction%
	\end{algorithmic}
\end{algorithm}

\section{Evaluation and Comparison}\label{sec: evaluation}

In this section, we evaluate PISTIS's reliability, latency, and
incurred overhead on network bandwidth.

\subsection{PISTIS's latency vs.\ related systems' latency}

We begin with a latency
comparison between PISTIS and other related works based on the worst case incurred delay. We compute worst case delays from the bounds established for each algorithm (a direct
experimental evaluation would not be fair, since not all previous
work~\cite{flaviu} consider probabilistic synchronous networks).
\newtext{Later sections provide an experimental
  comparison with RT-ByzCast~\cite{RT-ByzCast}, the system most
  related to ours.}
We elaborate in what follows on the computation of the worst case delays. First we refine the definition of~$d$ introduced in
Sec.~\ref{system model}. Let $d_n$ be the maximum network delay, and
$d_p$ be the maximum local processing time, which includes the cryptographic
operations overhead, such that $d$ can be decomposed as $d_p+d_n$.
Christian et al.~\cite{flaviu} compute the worst case delay as
$10*(f+2)*(n-1)*d_n$ where $f$ is the maximum number of faulty
processes, $n$ the total number of processes, and $d_n$ the network
delay. In this work, $d_p$ is equal to $10$.
Kozhaya et. al~\cite{RT-ByzCast} compute the worst-case delay as
$3*R*d$, where $R$ is the number of consecutive synchronous
communication rounds the same message gets disseminated
(time-triggered re-transmissions). PISTIS's worst case delay is proved
to be $3*\mathbb{T}$.
To ensure fairness and consistency with the latency experiments
presented below, we set $R=8$ and $\mathbb{T}=8d$.
However, due to PISTIS's signature management (see, for
  example, the optimizations described in
  Sec.~\ref{sec:optimizations}), PISTIS's worst case delay can be
  alternatively computed as $(3*8*d_n)+(2*N*d_p)$.  This is in part due
  to the fact that in PISTIS nodes
  avoid re-verifying already verified signatures.

Our results, shown in Table~\ref{fig:worst-case-latencies}, show that
PISTIS has the best worst case latencies of
all algorithms for $d_n=1ms$ (as mentioned above,
  in the first column $d_p=10$, while in the last two columns $d_p$ is such that
  $1<d_p<10$, and can be derived from the numbers provided in the
  table).
\begin{table}[!h]
  \normalsize
  \begin{center}
    \begin{tabular}{ l|c|c|c }
      & \cite{flaviu} & \cite{RT-ByzCast} & PISTIS \\\hline
      $N = 25$,  $f = 8$  & 2,400 ms  & 26 ms  & 25.6 ms\\ \hline
      $N = 50$,  $f = 16$ & 8,640 ms  & 70 ms  & 27 ms\\ \hline
      $N = 100$, $f = 33$ & 34,650 ms & 150 ms & 30 ms\\
    \end{tabular}
  \end{center}
  \caption{Worst case latencies}
  \label{fig:worst-case-latencies}
\end{table}

Two main observations can be made:
(1)~compared to the other protocols, PISTIS has superior performance
due to the fact that PISTIS is event triggered, utilizes fast
signature schemes, reduces the number of signatures created and
verified, sends fewer messages (which increase individual message
failures) and allows processes for fast detection of their tardiness;
and (2) PISTIS's expected performance in practice (see
Fig.~\ref{fig:duration1ms}) is significantly better than the worst
case delay bound reported in the table.

\subsection{Implementation Optimizations}
\label{sec:optimizations}

We implemented three optimizations to improve the performance of PISTIS (as described in
  Sec.~\ref{sec:et-rtbyzcast}).
(1)~If a process $p_i$ knows that some process $p_j$ has already
received $2f+1$ \emph{echo} signatures for some message~$m$, $p_i$
stops sending echoes related to $m$ to $p_j$.  Every process
implements this optimization by maintaining a list, say~$\mathcal{L}$,
that contains all the processes from which it has heard $2f+1$
signatures for a given message.
During a broadcast, a process diffuses a message to $\fanout$
processes at random among $\mathit{\Pi}\setminus\mathcal{L}$.
Processes do the same for \emph{deliver} messages.
(2)~Processes do not verify signatures that
they have already received.
(3)~Processes skip messages that only contain signatures
that were already received.

\subsection{Implementation Configuration and Settings}

We implemented PISTIS in C++ on the Omnet++ 5.4.1 network
simulator~\cite{omnet}. In order to accurately measure PISTIS's
communication overhead, we configure network links to have a
non-limiting 1Gbps throughput, and a communication latency of either
1ms or 5ms. We evaluated PISTIS's performance using two signature
schemes of similar security guarantees, and available in the OpenSSL
library~\cite{openssl}: RSA-2048 (i.e., 256 bytes long signatures) and
\newtext{ECDSA} with prime256v1 curves (i.e., 71 bytes long
signatures). We use broadcast messages of sizes equal to 1B and
1KB.

We run our simulations for systems with $\total\in\{25,49,73, 300\}$
processes \newtext{in fully connected networks}, and for several
values of $\fanout$, which is the number of processes each process
forwards a message $m$ to during diffusion.
We consider the probability of losing/omitting a message sent at any
point in time to be $i/10$, where $0\leq{i}\leq{9}$.

\subsection{PISTIS's Reliability}
\label{sec: process and system shutdown}

To assess PISTIS's reliability, we evaluate the probability
that a correct process enters the passive mode.  Such probability is a crucial measure: a process becoming passive may lead the system to shutdown and hence to stop delivering messages.  Namely, when $\total=3f+1$, a single correct process staying passive for long-enough can,
in the worst case (when $f$ Byzantine processes are not sending messages), %
leave $2f$ correct
processes, which would not be enough to gather quorums of size $2f+1$,
leading those $2f$ processes to also become passive.

For a given value of $\total$ and $p$, we invoke a broadcast at one of
the processes and record any non-Byzantine process that
crashed itself during broadcast.
We obtain our results by repeating each experiment $10^5$~times, and
we report the probability that a process crashes itself as:
$$(\text{num. of experiments with self-crashed processes})/{10^5}$$

We study the impact of several parameters, including
$\mathbb{T}$, $\total$, $\fanout$, $f$, and $p$, on PISTIS's
reliability, and determine which values should be used to 
enforce an intended system reliability.

Fig.~\ref{pcrashchangingf68} shows that the system's reliability increases
with its size
and $\duration$'s value for large enough values of $\mathbb{T}/d$.
For example, when $\mathbb{T}=8d$, a system with 25
(resp.\ 49) processes operates with high reliability (i.e., there
is a negligible probability that a process becomes passive) under message
loss rates reaching up to
 $40\%$ (resp.\ $50\%$).

Fig.~\ref{pcrashchangingf2} shows that the
actual number of Byzantine processes, which varies between $0$ and $f$
(the maximum number of tolerable Byzantine nodes), influences the
system's resiliency.
As expected, with fewer processes being Byzantine, higher message
loss rates are tolerated without any process shutdown.

\medskip
\textbf{Impact of the diffusion fanout}.
In the results presented so far, processes forward each message
to $\fanout=f+1$ other random processes.  We now study the effect of
$\fanout$ by measuring PISTIS's reliability when it
varies.
Fig.~\ref{fig:changingX} shows that increasing $\fanout$ helps increase
the overall system reliability. As expected increasing the fanout (value of $\fanout$) reduces the probability of having a non-Byzantine node becoming passive.

	\begin{figure}[!t] %
		\begin{center}
			\includegraphics[width=.8\columnwidth]{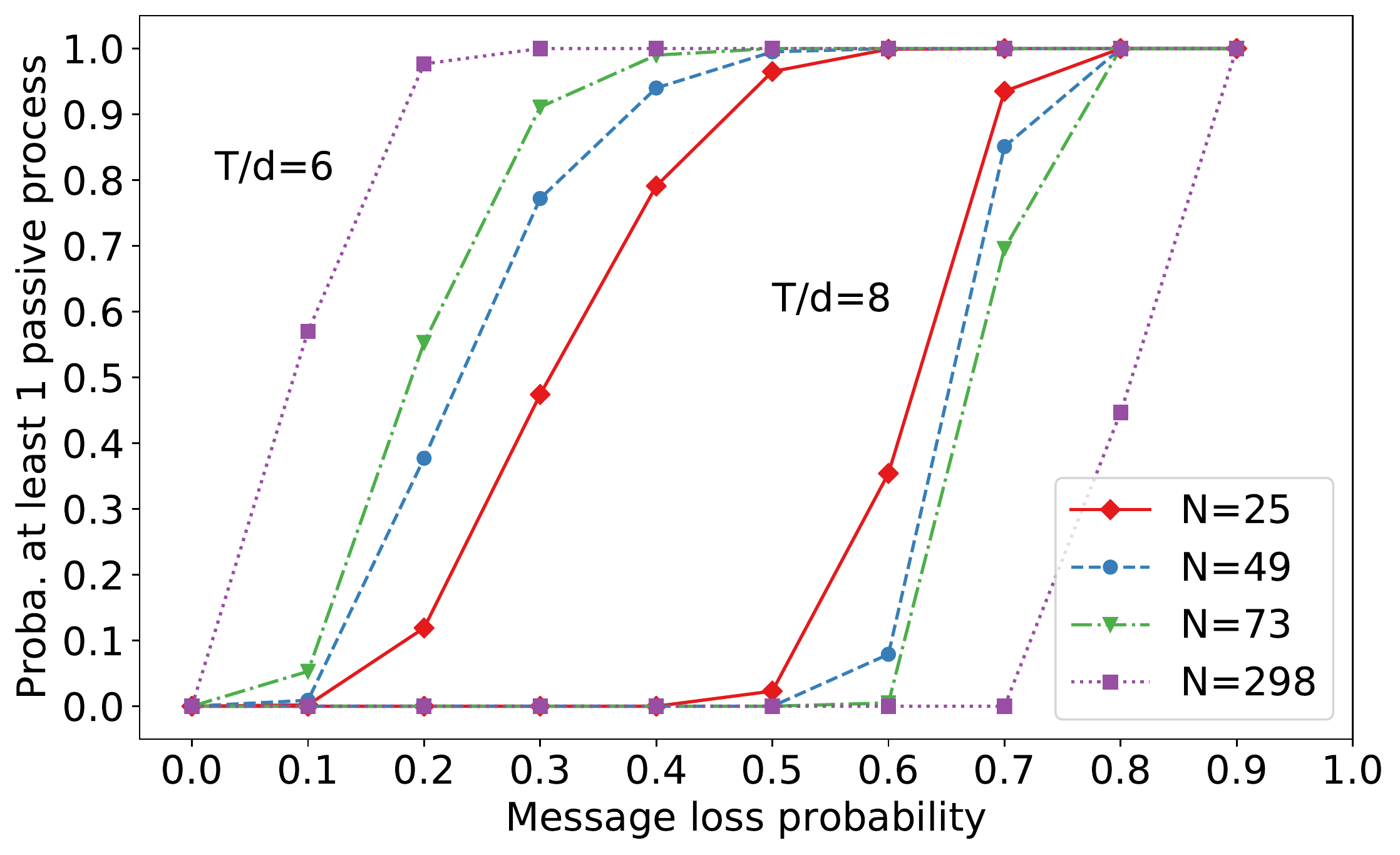}
			\vspace*{-10pt}
			\caption{Probability of a correct process
                          becoming passive when $\mathbb{T}=6d$ or
                          $\mathbb{T}=8d$, and $\fanout=f+1$ (without recovery)}\label{pcrashchangingf68}
			\vspace*{-10pt}	
		\end{center}
	\end{figure}

	\begin{figure}[!t] %
		\begin{center}
			\includegraphics[width=.8\columnwidth]{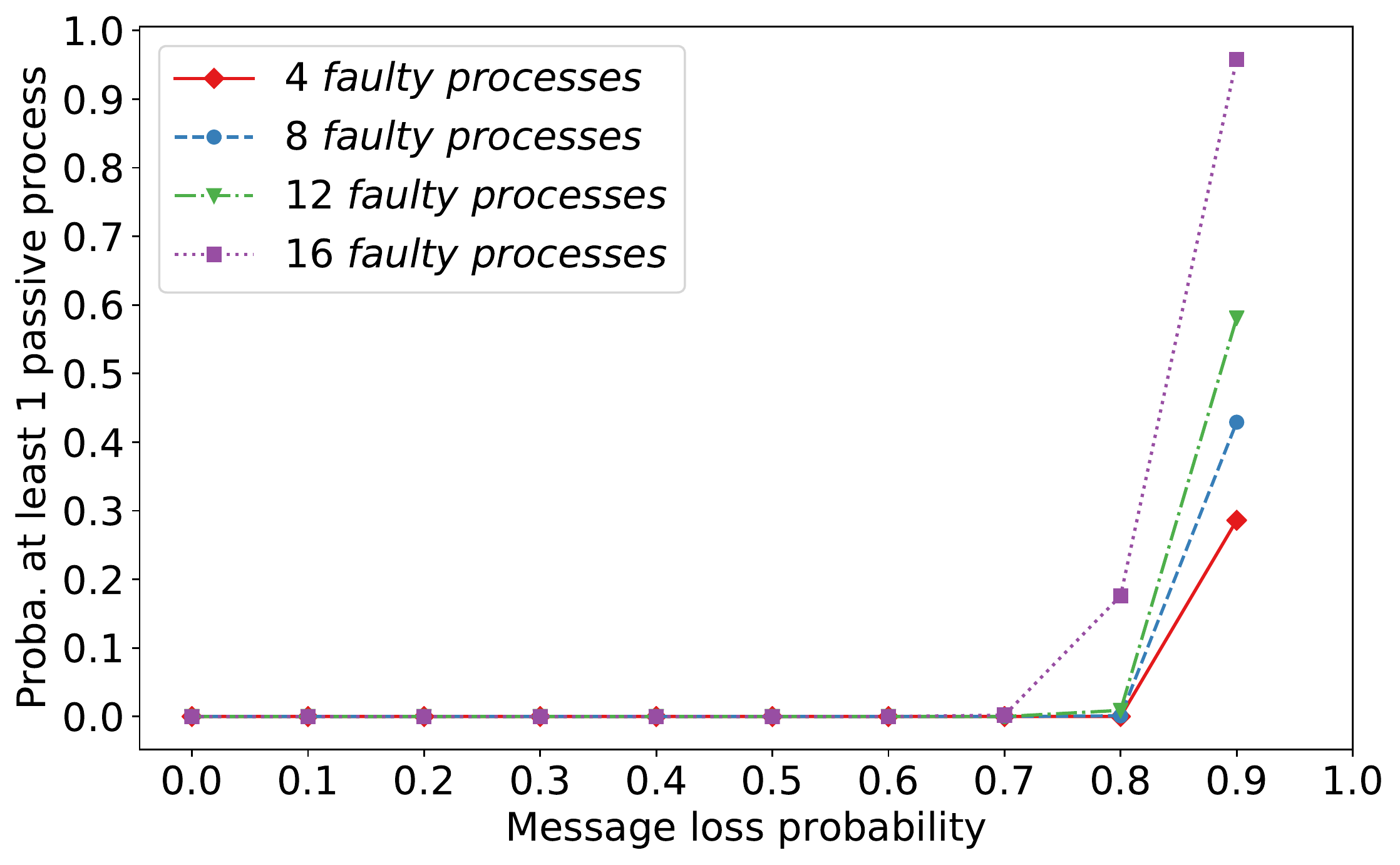}
			\vspace*{-10pt}
			\caption{Probability of a correct process becoming passive in a system of $49$ processes (i.e., $f=16$) using $\mathbb{T}=8d$ and $X=17$, when 0, 4, 8, 12 or 16 processes are faulty (without recovery)}\label{pcrashchangingf2}
			\vspace*{-15pt}
		\end{center}
	\end{figure}
	\begin{figure}[!t] %
		\begin{center}
			\includegraphics[width=.8\columnwidth]{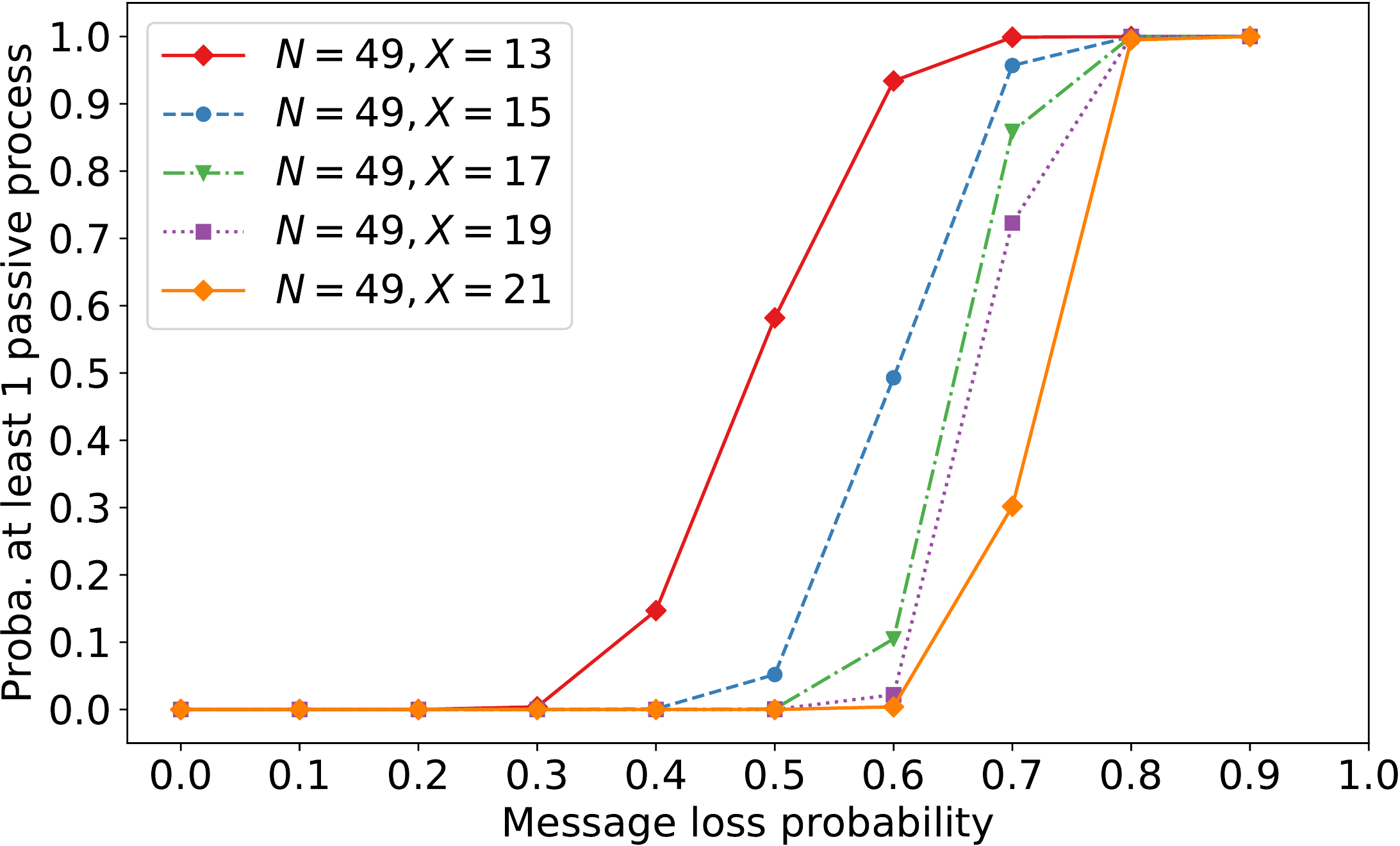}
			\vspace*{-10pt}
			\caption{Probability of a correct process becoming passive in a system of $49$ processes using $\mathbb{T}=8d$, and where $\fanout$ varies (without recovery)}\label{fig:changingX}
			\vspace*{-15pt}
		\end{center}
	\end{figure}

\medskip
\textbf{Recovery.}
Fig.~\ref{fig:recovery} details the probability that no Byzantine
quorum remains active after a broadcast instance when the message loss
probability increases. First, one can observe that the recovery
mechanisms improve the resiliency of the system. For example, with
$N=49$, PISTIS can tolerate a 70\% message loss rate without
system-wide crashes thanks to the recovery mechanisms, improving over
the value of 50\% obtained without recovery. Second, we show that one
can further improve the system's tolerance to message losses by
overprovisioning the system. By using three more nodes, i.e., 52 in
total, the system can tolerate $f=16$ Byzantine nodes and now tolerate
up to 80\% of message losses.

\begin{figure}[!t]%
	\begin{center}
		\includegraphics[width=.8\columnwidth]{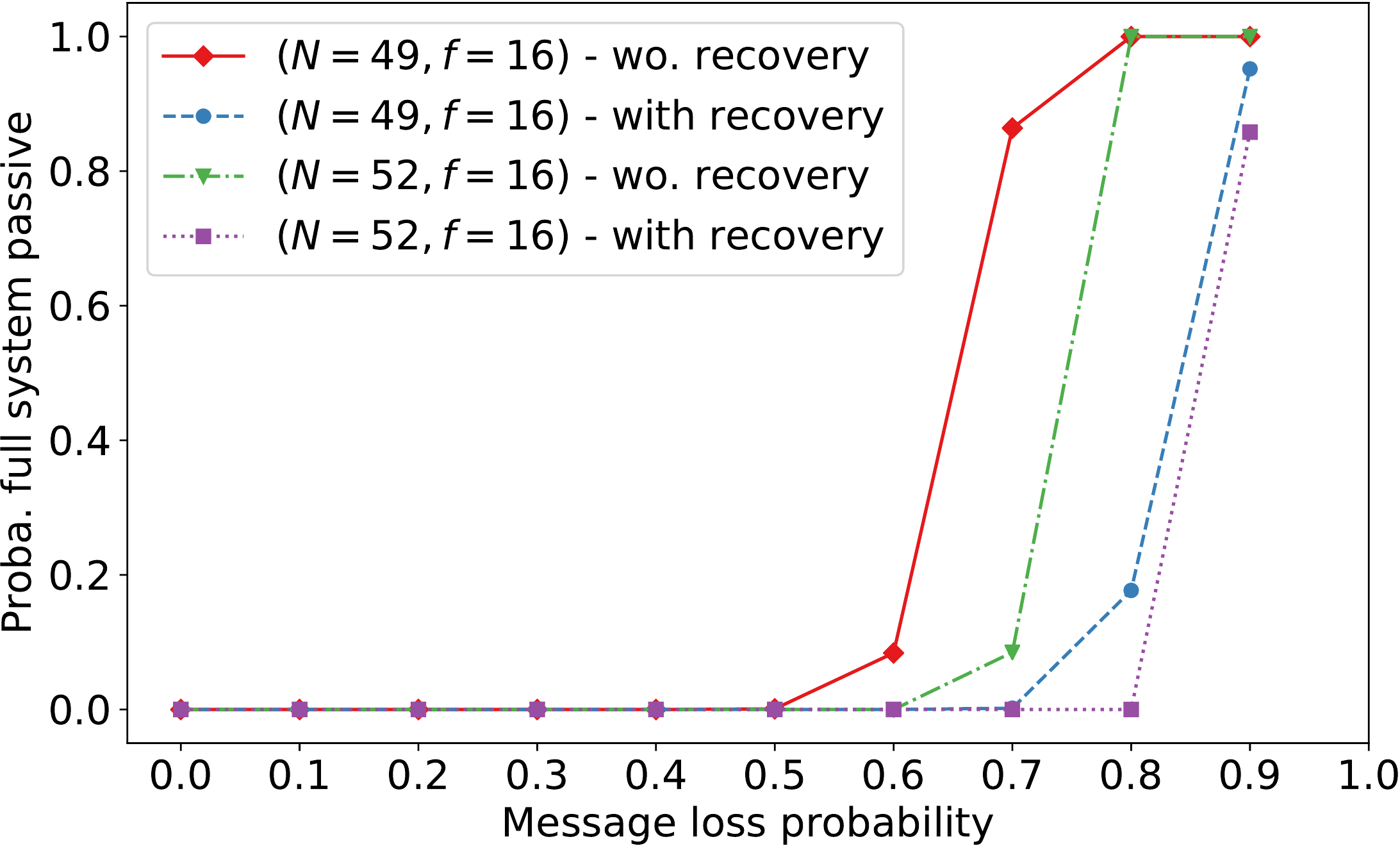}
                \vspace*{-10pt}
                \caption{Probability that no Byzantine quorum remains active in systems of 49 or 52 processes, when $\mathbb{T}=8d$, $\fanout = 17$, and $f=16$ processes are Byzantine.}
                \label{fig:recovery}
                \vspace*{-15pt}
	\end{center}
\end{figure}

\subsection{PISTIS latency and bandwidth consumption}
\label{sec:PISTIS-overhead}

Next, we evaluate PISTIS's incurred bandwidth and latency. For these experiments,
we average results over 1,000 runs. %
We use $\mathbb{T}=8d$, since our reliability results
show it allows a very large number of
message losses to be tolerated. However, we now run our experiments without any
message losses to measure the worst case bandwidth consumption. We
measure both the protocol latency and bandwidth consumption depending
on the value of $\fanout$ that the processes use.
\newtext{We also compare the average latency and bandwidth consumption
  of PISTIS with that of RT-ByzCast~\cite{RT-ByzCast}. Note that RT-ByzCast~\cite{RT-ByzCast} uses ECDSA signatures and all-to-all communication ($\fanout=N$).} %

\medskip
\textbf{Latency.} Fig.~\ref{fig:duration1ms} and~\ref{fig:duration5ms} detail the
latency for a broadcast message to be delivered by all correct
processes in systems of size 25, 49, and 73 \newtext{(i.e., where $f \in \{8, 16, 24\}$)}:
PISTIS delivers with latencies within $[3\text{ms},60\text{ms}]$ depending on the network delay $d$ and signature scheme used RSA vs. ECDSA.  The
latency increases when $\total$ increases, and decreases when
$\fanout$ increases.
We draw the following conclusions:
(1)~PISTIS is slower than
RT-ByzCast for $\fanout<f$. For $\fanout \geq f$ PISTIS is on a par with RT-ByzCast until some $\fanout \leq 3f$ ($\fanout\leq 2f$ for systems with up to 400 nodes, see Table II)  after which PISTIS is faster;
(2)~PISTIS's absolute improvement over RT-Byzcast becomes more significant with
increased link delay;
(3)~When delivering latencies on par with or better than RT-ByzCast, PISTIS can do so
with a lower network overhead as presented next (see Fig.~\ref{fig:bdw1ms}
and~\ref{fig:bdw5ms}).

\begin{figure}[!t] %
	\begin{center}
		\includegraphics[width=.8\columnwidth]{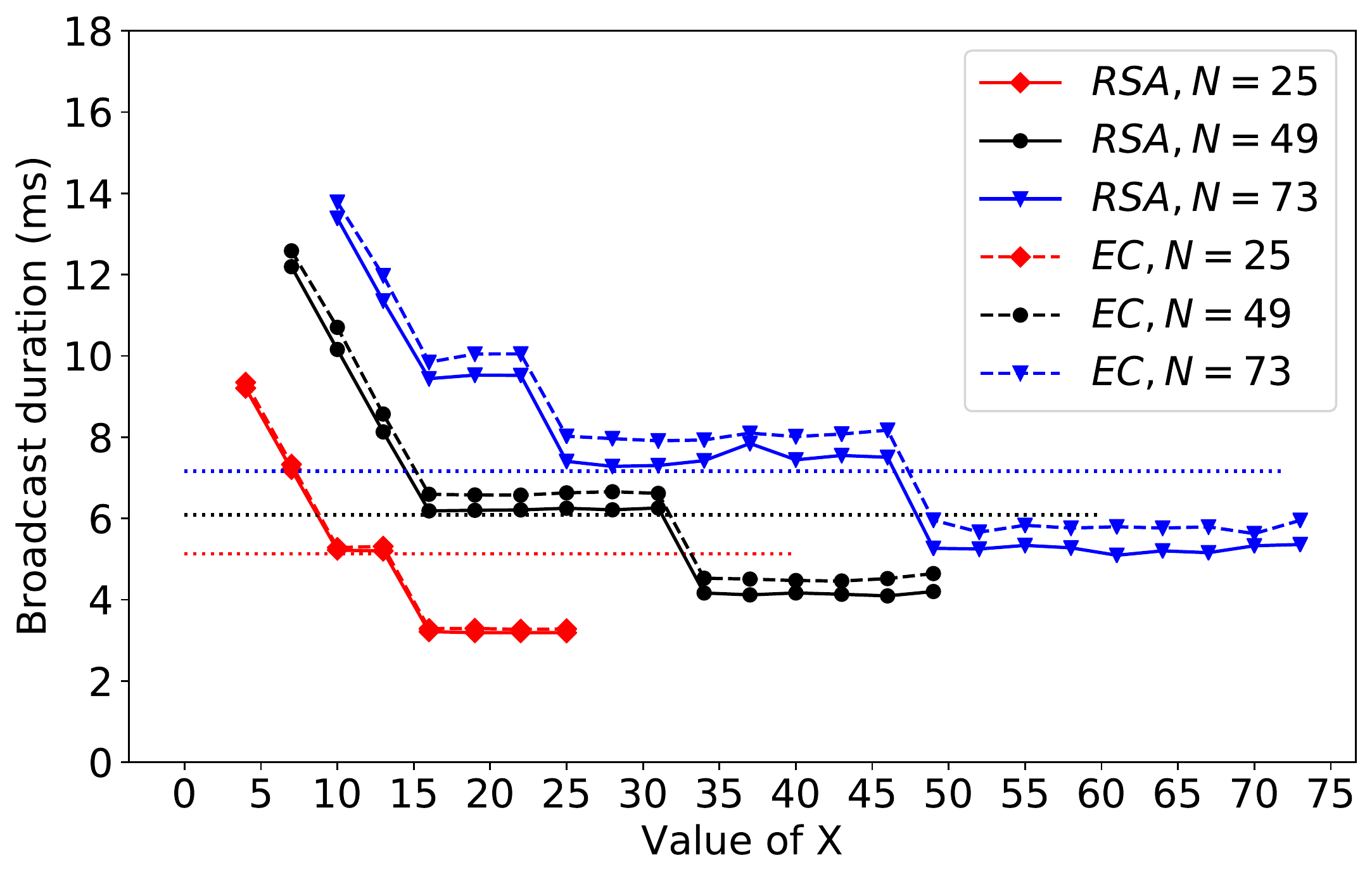}
                \vspace*{-10pt}
		\caption{Average latency with a 1ms link latency with $\mathbb{T}=8d$ and without message losses. The dotted lines indicate RT-ByzCast's values~\cite{RT-ByzCast}.}
                \label{fig:duration1ms}
		\vspace*{-15pt}
	\end{center}
\end{figure}

\begin{figure}[!t] %
	\begin{center}
		\includegraphics[width=.8\columnwidth]{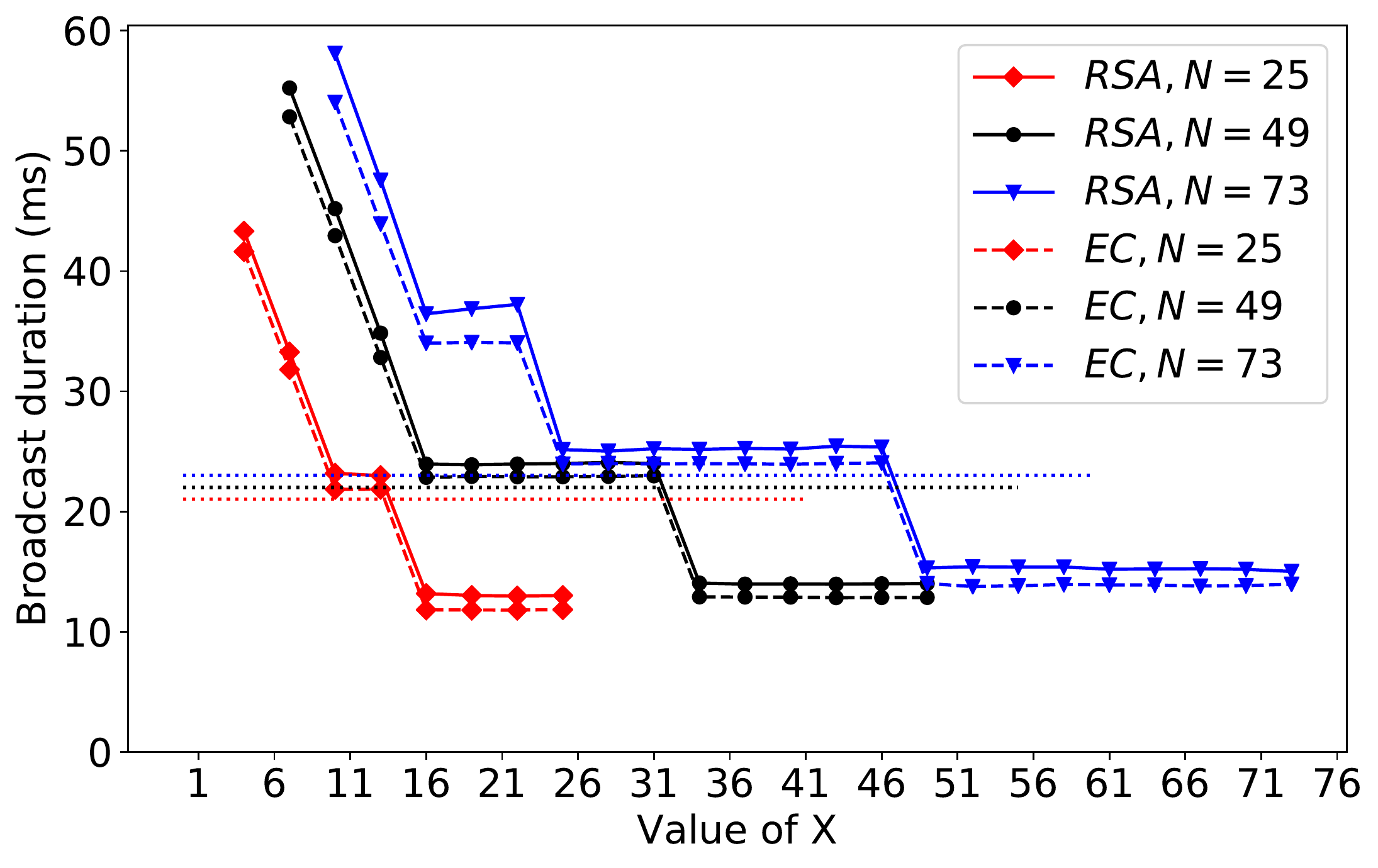}
		\vspace*{-10pt}
		\caption{Average latency with a 5ms link latency. The dotted lines indicate RT-ByzCast's values~\cite{RT-ByzCast}.}
                \label{fig:duration5ms}
		\vspace*{-15pt}
	\end{center}
\end{figure}

\medskip
\textbf{Network bandwidth consumption.}
We now measure PISTIS's bandwidth overhead per broadcast
invocation, using RSA and ECDSA signatures.
Fig.~\ref{fig:bdw1ms} and~\ref{fig:bdw5ms} present the bandwidth
consumption for 1B payloads with 1ms and 5ms link delay,
respectively. One can observe that
with
$\fanout=f+1$ and when using ECDSA signatures, PISTIS's bandwidth
consumption is 3.2 times lower than that of RT-ByzCast. %
\newtext{We also observe that when using ECDSA
  signatures there is a fanout between $2f+1$ and $3f+1$ such that
  below this fanout PISTIS's average bandwidth consumption is lower
  than RT-ByzCast's, while past that threshold, PISTIS's average
  bandwidth consumption becomes greater than RT-ByzCast's. This is
  partly due to the fact that PISTIS being event-based sometimes
  consumes more bandwidth. However, we see in those figures that
  PISTIS provides a useful trade-off between latency and bandwidth
  consumption.} Fig.~\ref{fig:bdw1K} shows as well that the bandwidth
consumption increases
 reasonably when the message payload is
increased to \newtext{1KB}. Besides bandwidth,
  \ifx\tAppA\tAppB%
  Fig.~\ref{fig:msgs} (Appx.~\ref{sec:appendixmsgs})
  \else%
  Fig.~12 in~\cite[Appx.E]{Kozhaya+Decouchant+Rahli+Verissimo:pistis:long:2019}
  \fi
 shows that PISITS also sends less message than RT-ByzCast.%

\begin{figure}[!t]%
	\begin{center}
		\includegraphics[width=.8\columnwidth]{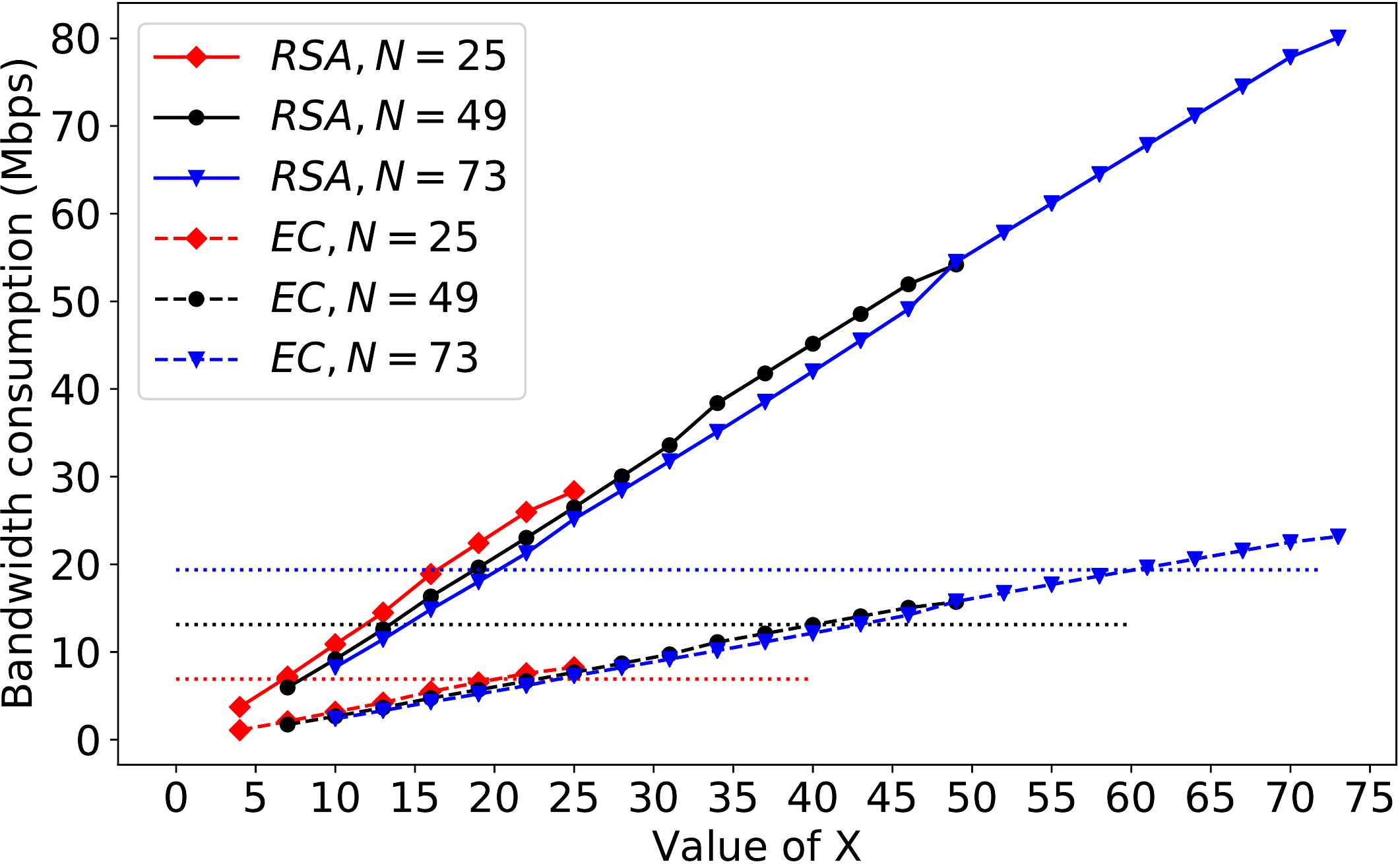}
                \vspace*{-10pt}
                \caption{Average bandwidth consumption per node and per communication link with a 1ms link latency without message losses. The dotted lines indicate RT-ByzCast's values~\cite{RT-ByzCast}.}\label{fig:bdw1ms}
                \vspace*{-15pt}
	\end{center}
\end{figure}

\begin{figure}[!t]%
	\begin{center}
		\includegraphics[width=.8\columnwidth]{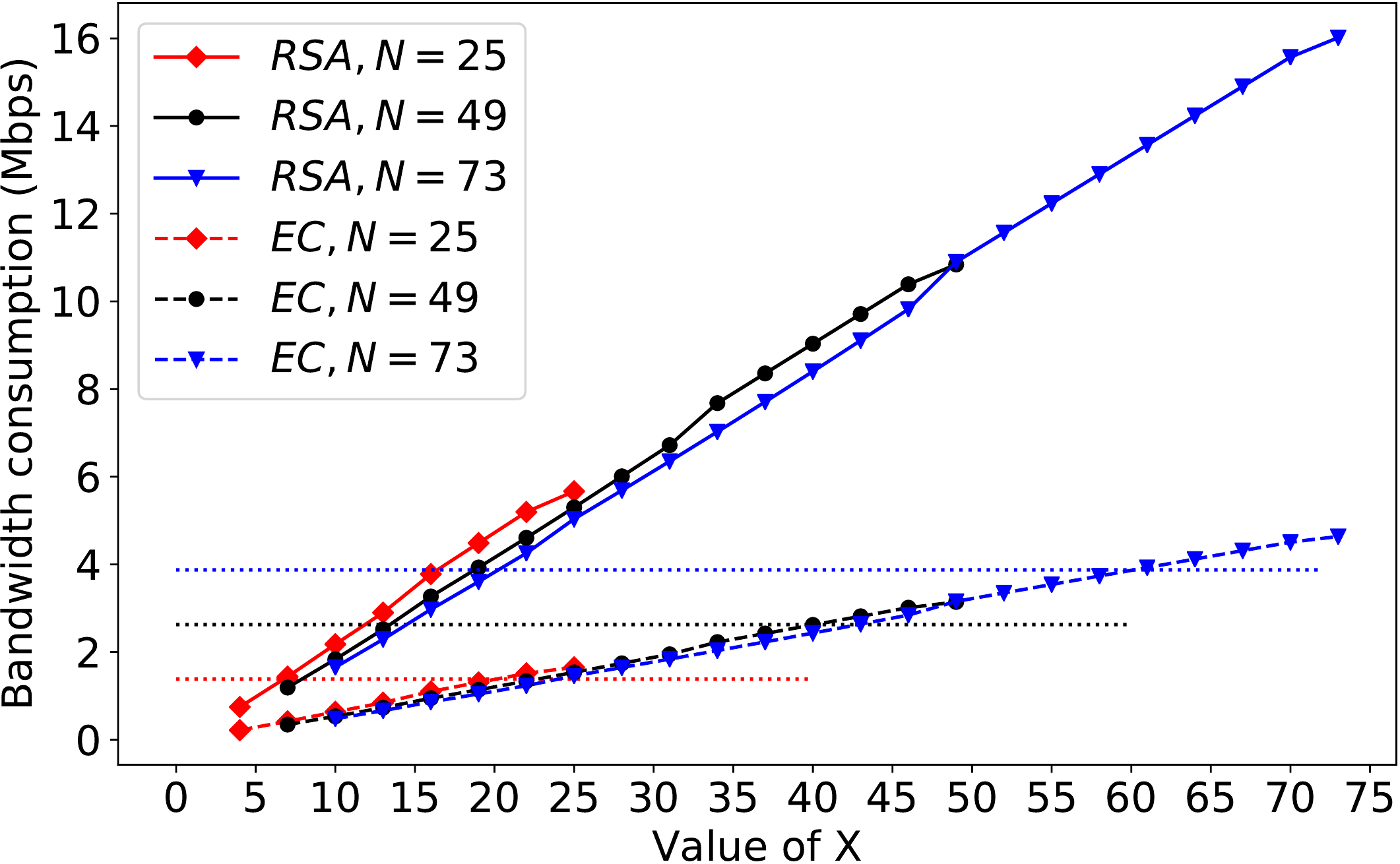}
                \vspace*{-10pt}
                \caption{Average bandwidth consumption per node and per communication link with a 5ms link latency without message losses. The dotted lines indicate RT-ByzCast's values~\cite{RT-ByzCast}.}\label{fig:bdw5ms}
                \vspace*{-15pt}
	\end{center}
\end{figure}

\begin{figure}[!t]%
	\begin{center}
		\includegraphics[width=.8\columnwidth]{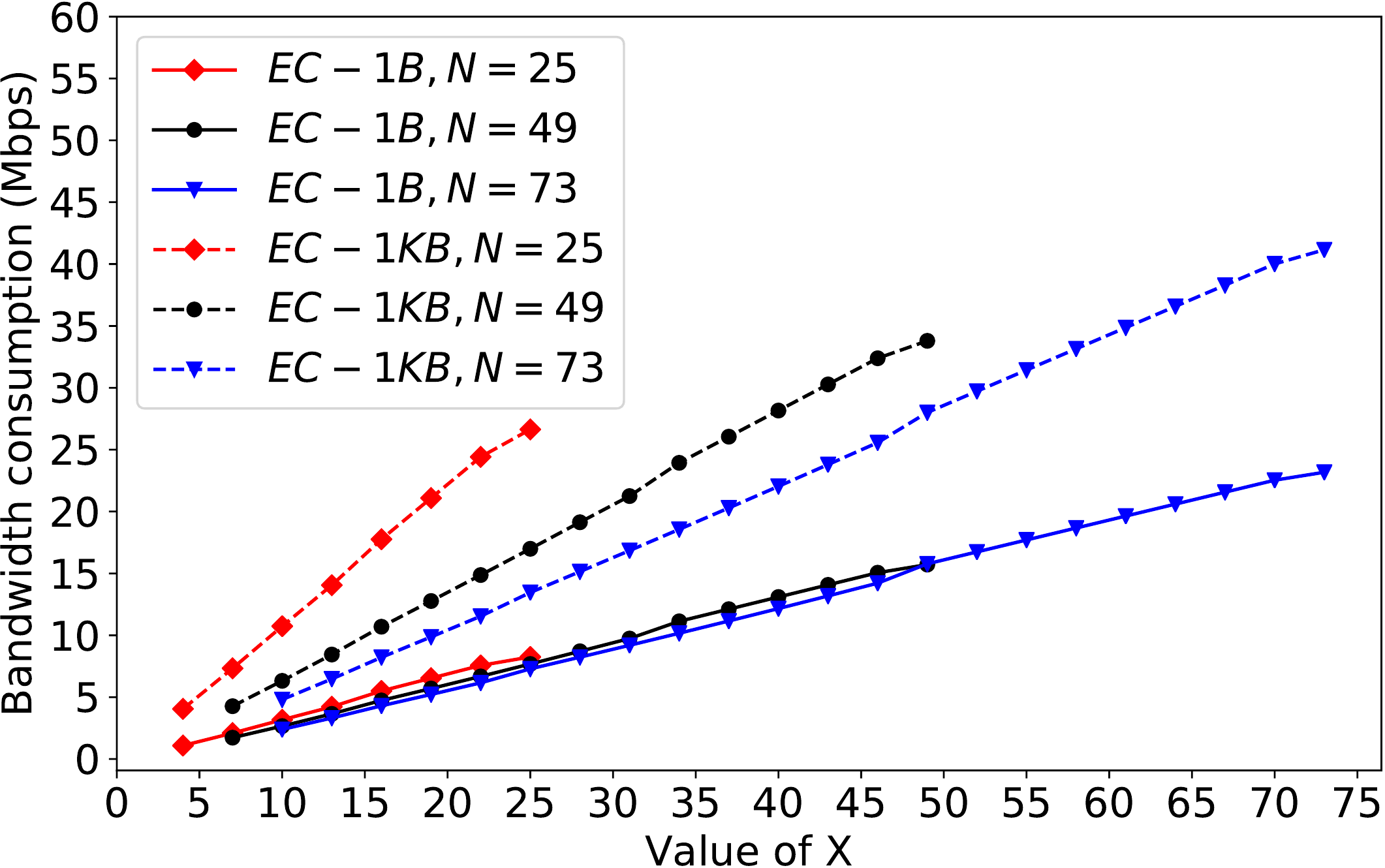}
                \vspace*{-10pt}
                \caption{Average bandwidth consumption per node and per communication link with a 1ms link latency using either 1B or 1KB messages, without message losses}\label{fig:bdw1K}
                \vspace*{-15pt}
	\end{center}
\end{figure}

\begin{table}[!h] \label{table:compare}
  \small
  \begin{center}
    \begin{tabular}{p{.46cm}|p{.6cm}|p{.6cm}|p{.6cm}|p{.55cm}||p{.6cm}|p{.55cm}|p{.6cm}|p{.5cm}}
      N & Bdw, \newline $X_{\mathit{min}}$ & Bdw, \newline $X_{\mathit{mid}}$ & Bdw, \newline $X_{\mathit{max}}$ & Bdw \cite{RT-ByzCast} & Lat, \newline $X_{\mathit{min}}$ & Lat, \newline $X_{\mathit{mid}}$ & Lat, \newline $X_{\mathit{max}}$  & Lat \cite{RT-ByzCast}\\\hline
      25   & 0.6   & 1.2  & 1.7    &  1.4   & 21.1  & 11.0 & 11.1    &  20.9  \\ 
      49   & 1.0   & 2.2  & 3.1    &  2.6   & 22.3  & 12.4 & 12.0    &  22.0  \\
      73   & 1.5   & 3.2  & 4.6    &  3.9   & 23.6  & 13.1 & 13.2    &  23.1  \\
      200  & 3.8   & 8.4  & 12.5   &  10.4  & 31.5  & 20.7 & 19.7    &  29.3  \\
      300  & 5.7   & 12.5 & 18.6   &  15.6  & 41.2  & 31.2 & 27.4    &  38.0  \\ 
      400  & 7.6   & 16.7 & 25.0   &  20.9  & 59.7  & 43.0 & 32.0    &  41.2  \\ 
      500  & 9.4   & 20.8 & 31.1   &  26.0  & 85.1  & 63.0 & 40.0    &  51.6  \\ 
      1000 & 18.7 & 41.4 & 62.2 &  52  & 296.3 &  213.1 &  98.5 &  116.2 \\
    \end{tabular}
  \end{center}
  \caption{Pistis bandwidth consumption (Mbps) and broadcast duration
    (ms) with larger systems ($f = \lfloor N/3 \rfloor$), where
    $X_{\mathit{min}}=f+1$, $X_{\mathit{mid}}=2f+1$ and $X_{\mathit{max}}=N$.}
  \label{fig:largeSystemsNumbers}
\end{table}

\medskip
\newtext{
\textbf{Scalability with the system size.}
We also evaluated how PISTIS' latency and bandwidth consumption evolve
with larger system sizes, namely up to 1000 nodes for $X\geq f+1$ and
a 5ms link latency.
Table~\ref{fig:largeSystemsNumbers} summarizes the results obtained
for $X=f+1$,  $X=2f+1$ and $X=N$.
Our results show that PISTIS outperforms RT-ByzCast and provides latencies suitable for (1) fast automatic interactions
($\leq{20}\mbox{ms}$) for systems with up to 200 nodes, (2) power systems and substation
automation
applications ($\leq{100}\mbox{ms}$) for systems with up to 1000 nodes, and (3) slow speed
auto-control functions ($\leq{500}\mbox{ms}$), continuous control
applications ($\leq{1}\mbox{s}$) and
 operator commands of SCADA
applications ($\leq{2}\mbox{s}$) for systems with 1000 nodes or more.
}

\section{Conclusion}
\label{conclusion}

In this paper, we studied how to build large-scale distributed
protocols that tolerate network faults and attacks while providing
real-time communication. We introduced a suite of proven correct
algorithms, starting from a baseline real-time Byzantine reliable
broadcast algorithm, called PISTIS, all the way up to real-time
Byzantine atomic broadcast and consensus algorithms. PISTIS is
empirically shown to be robust, scalable, and capable of meeting
timing deadlines of real CPS applications. PISTIS withstands message
loss (and delay) rates up to
50$\%$
 in systems with 49 nodes and
provides bounded delivery latencies in the order of a few
milliseconds. PISTIS improves over the state-of-the-art in scalability
and latency through its event-triggered nature, gossip-based
communications, and fast signature verifications. Our work simplifies
the construction of powerful distributed and decentralized monitoring
and control applications of various CPS domains, including
state-machine replication for fault and intrusion tolerance.

\printbibliography

\vspace*{-30pt}
\renewcommand{\baselinestretch}{0.85}
\begin{IEEEbiography}[{\includegraphics[width=1in,height=1.25in,clip,keepaspectratio]{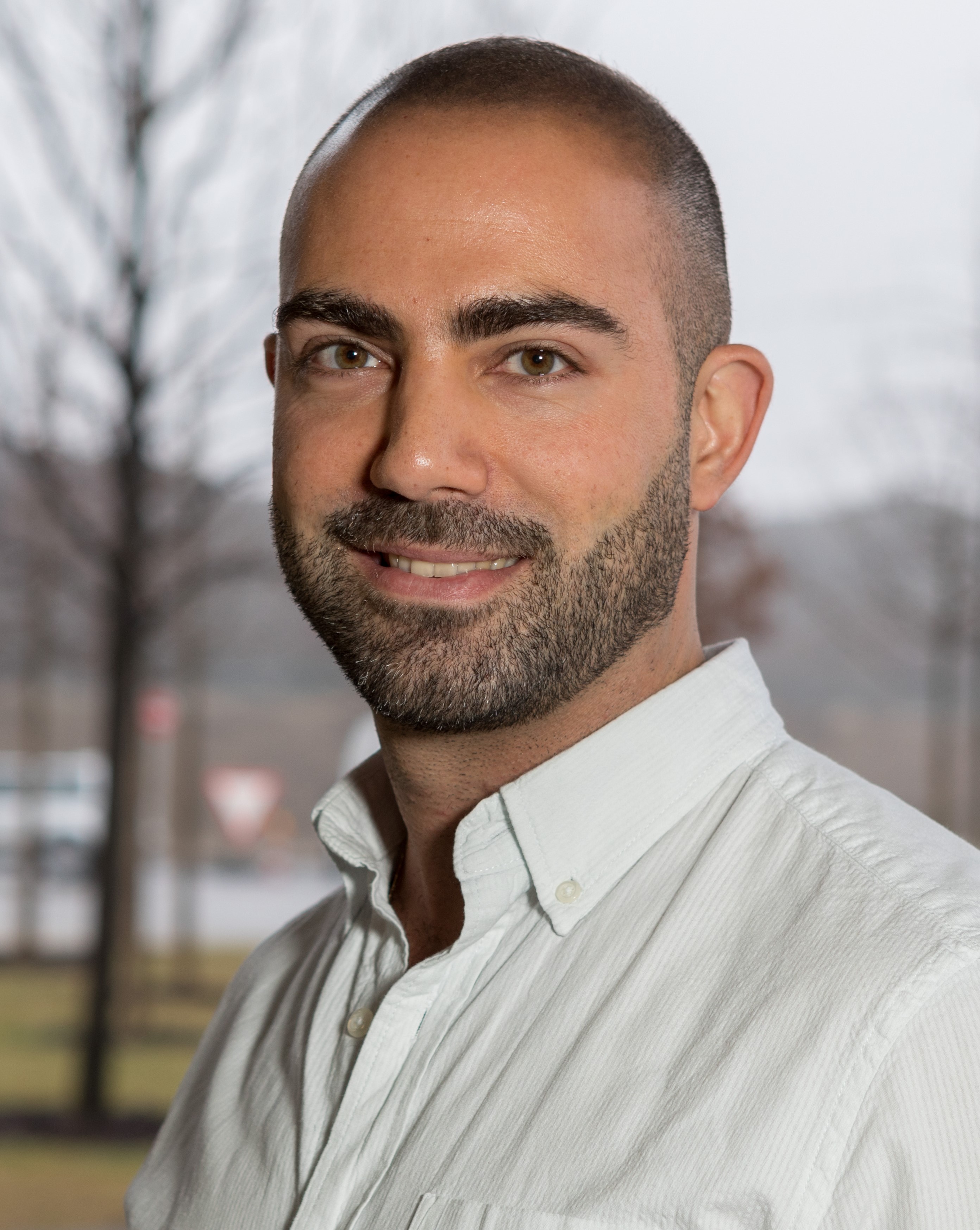}}]{David Kozhaya} %
	is a Senior Scientist at ABB Research, Switzerland.
	He received his PhD degree in Computer Science in 2016, from EPFL, Switzerland, where he was granted a fellowship from the doctoral school. His primary research interests include reliable distributed computing, real-time distributed systems, and fault- and intrusion-tolerant distributed algorithms.%
\end{IEEEbiography}

\vspace*{-30pt}
\begin{IEEEbiography}[{\includegraphics[width=1in,height=1.25in,clip,keepaspectratio]{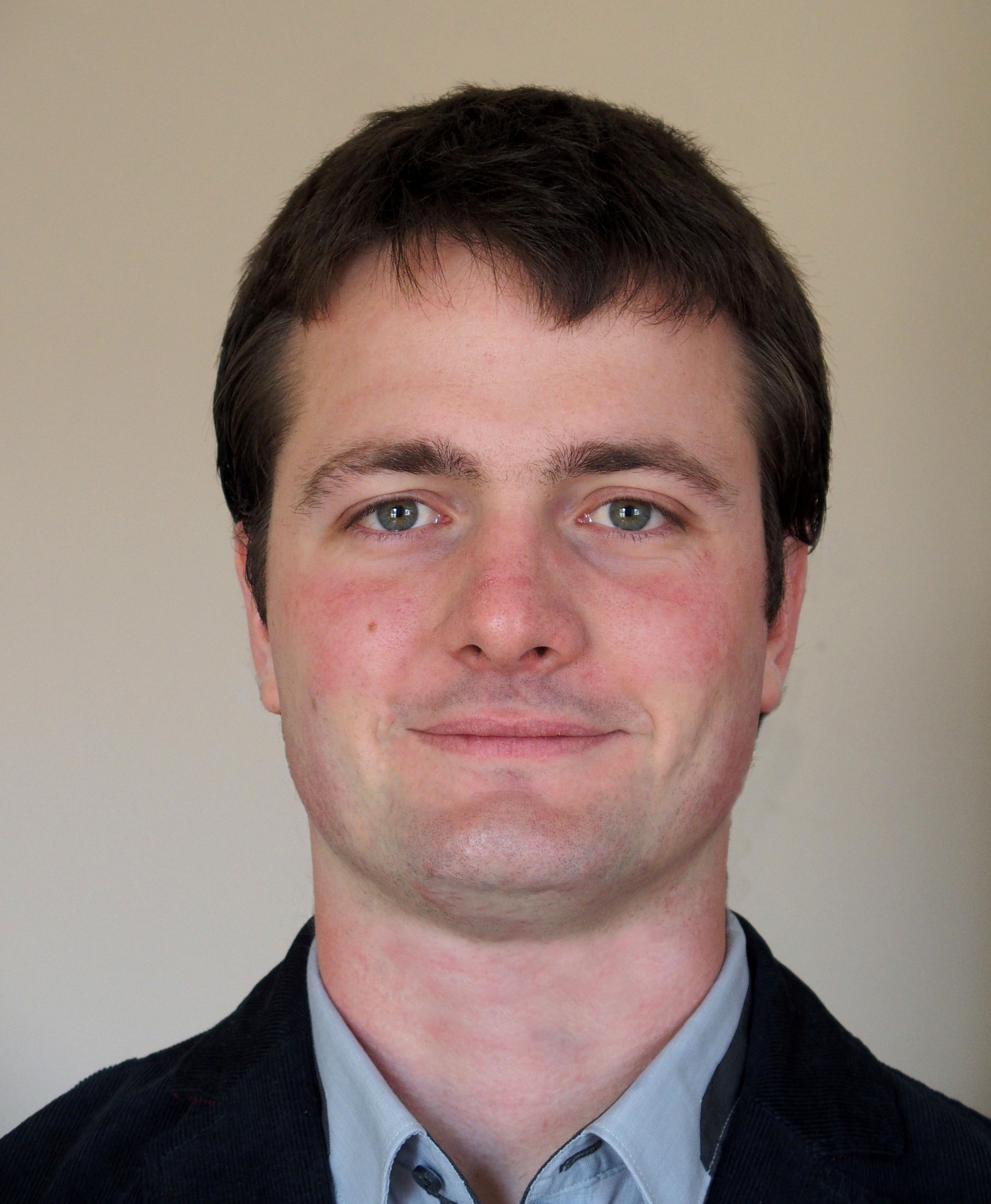}}]{J\'er\'emie Decouchant} is an Assistant Professor at TU Delft, the Netherlands. 
He received his Ph.D in Computer Science in 2015 from the Grenoble-Alpes University, France. 
His research interests include resilient distributed computing, privacy-preserving systems,  and their application to Blockchain, genomics, and machine learning.
\end{IEEEbiography}
\vspace*{-30pt}
\begin{IEEEbiography}[{\includegraphics[width=1in,height=1.25in,clip,keepaspectratio]{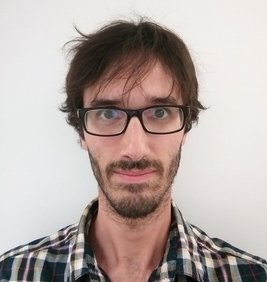}}]{Vincent Rahli} is a
Senior Lecturer at the University of Birmingham. He received his Ph.D
in Computer Science from Heriot-Watt University, UK. His research
focuses on designing, formalizing, and using type theories and on the
verification of distributed systems using proof assistants.
\end{IEEEbiography}

\vspace*{-30pt}
\begin{IEEEbiography}[{\includegraphics[width=1in,height=1.25in,clip,keepaspectratio]{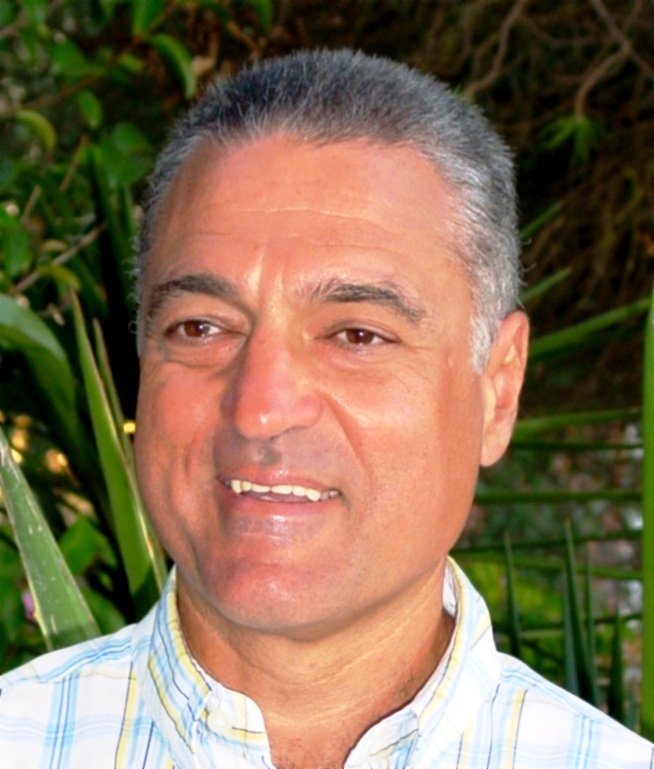}}]{Paulo Esteves-Ver\'issimo} is a professor at the KAUST University (KSA), and Director of the Resilient Computing and Cybersecurity Center (RC3 - \url{https://rc3.kaust.edu.sa/}). He was a member of the Sci\&Tech. Comm. of ECSO EU Cyber Security Organisation, Chair of IFIP WG 10.4 on Dependable Comp. and F/T, and vice-Chair of the Steer. Comm. of the DSN conference. He is Fellow of IEEE and of ACM, and associate editor of the IEEE TETC journal, author of over 200 peer-refereed publications and co-author of 5 books. He is currently interested in resilient computing, and its potential to improve classic cybersecurity techniques, in: SDN-based infrastructures; autonomous vehicles from earth to space; distributed control systems; digital health and genomics; or blockchain and cryptocurrencies. 

\end{IEEEbiography}

\ifx\tAppA\tAppB%

\newpage

\appendices

\begin{refsection}

\section{Differences Between Probabilistic Synchrony and Other Standard Models}
\label{app:partial-sync}

\paragraph{Comparison with fully asynchronous models}
Our model is more
informative than traditional fully asynchronous models. More precisely,
asynchronous models do not make any assumptions
regarding message transmission and processing delays, while we assume
that messages are delivered within a maximum transmission delay $d$
with high probability.

\paragraph{Comparison with synchronous models}
Our communication model is a probabilistic synchronous one. We recall
that in every transmission attempt a link may (with some probability)
violate reliability and timeliness by dropping the message or
delivering it within a delay~$>d$. In case of message loss (omission)
a sender that needs to re-transmit that message again faces yet
another risk of transmission failure. Due to omissions (losses in
consecutive transmission attempts) and the required follow-up
re-transmissions, the time it takes to send a message reliably from
one process to another (measured from the time of the first
transmission attempt) may be unbounded. So, despite links being
reliable and timely with high probability, our communication system is
no longer synchronous.

\paragraph{Comparison with partially synchronous models}
In comparison with partial synchrony~\cite{Dwork:1987}, which assumes
that communication becomes forever synchronous after some unknown
point in time, our probabilistic synchronous model
guarantees only finite synchronous
periods (with variable durations) that may occur randomly during the
lifetime of the system. In fact such probabilistic synchronous
communication has been shown to be weaker, in some
sense~\cite{neverSayNever:ipdps}, than partial synchrony. For example,
while the celebrated failure detectors of~\cite{Chandra:unreliableFD}
can be implemented in partially synchronous systems they are
impossible to implement in the systems with probabilistic synchronous
communication~\cite{neverSayNever:ipdps}.

\paragraph{The need for probabilistic synchrony models}
Probabilistic synchronous models (such as the one presented here
or~\cite{Verissimo+Almeida:tcos:1995,neverSayNever:ipdps}) are more
``realistic'' than synchronous models in the sense that timing
assumptions cannot always be ensured in distributed systems because,
for example, of the difficulty of guaranteeing reliable communication
between the nodes of a system.
Making the probability of timing failures (e.g., that messages might
be delivered after $d$) transparent to the model and protocols makes
them more robust.
For example, it allows designing protocols where messages might not
always arrive within a specified maximum transmission delay. Systems
that require processes to operate in a timely fashion, such as mission
critical systems, can therefore dynamically adapt to such untimely
situations to ensure that timing guarantees are fulfilled.

\paragraph{Comparison with quasi-synchronous models}
Quasi-synchronous models~\cite{Verissimo+Almeida:tcos:1995} address
the timing issues mentioned above.
In~\cite{Verissimo+Almeida:tcos:1995} synchronism is characterized by
the following properties:
P1---processing speeds are bounded and known;
P2---message delivery delays are bounded and known;
P3---local clock rate drifts are bounded and known;
P4---load patterns are bounded and known;
and P5---differences among local clocks are bounded and known.
A system is quasi-synchronous if it satisfies properties P1--P5, and
at least one of those does not hold with some known non-zero
probability.
As in a quasi-synchronous model, in our probabilistic model P2 only
holds with high probability.
Note, however, that in our probabilistic model we do not assume that
differences among local clocks are bounded and known.

\section{Correctness of PISTIS (Algorithm~\ref{main-alg})---Proof of Theorem~\ref{theorem:correctness-pistis}}
\label{appx:correctness-pistis}

\begin{lemma}[Validity]
  \label{lem:pistis-validity}
  If a correct process $p_i$ broadcasts $m$ then $p_i$ eventually
  delivers $m$.
\end{lemma}

\begin{proof}[Proof outline]
  Because $p_i$ is correct, it will hear echoes of $m$ from $2f+1$
  processes (including $p_i$) by $t+\mathbb{T}$, where $t$ is the time
  $p_i$ broadcasted $m$. This is true as otherwise, i.e., if less than
  $2f+1$ echoes for $m$ are heard, $p_i$ would kill itself (hence is
  no longer correct).
  Indeed, $p_i$ triggered a timer (see line~\ref{alg:timeout} of
  Algorithm~\ref{main-alg}) when it started broadcasting $m$ (see
  line~\ref{alg:diffuse-broadcast}).
  Because $p_i$ received $2f+1$ echoes for $m$,
  it must have delivered $m$ too (see
  lines~\ref{alg:line:delmsga},~\ref{alg:line:delmsgb},
  and~\ref{alg:line:delmsgc} of Algorithm~\ref{main-alg}).
\end{proof}

\begin{lemma}[No duplication]
  No correct process delivers message $m$ more than once.
\end{lemma}

\begin{proof}[Proof outline]
  According to line~\ref{alg:line:delmsgfun}  of Algorithm~\ref{main-alg}
  a process only delivers a message if the corresponding
  $\rdelivSYMB$ does not exist, and creates one right after
  delivering, thereby preventing from delivering a message twice.
\end{proof}

\begin{lemma}[Integrity]
  If some correct process $p_j$ delivers a message $m$ with correct
  sender $p_i$, then $m$ was previously broadcasted by $p_i$.
\end{lemma}

\begin{proof}[Proof outline]
  Because $p_j$ delivered $m$, it must have received $2f+1$ signed
  echoes for $m$ (see
  lines~\ref{alg:line:delmsga},~\ref{alg:line:delmsgb},~\ref{alg:line:delmsgc},
  and~\ref{alg:line:delmsgd} of Algorithm~\ref{main-alg}).  As
  mentioned in Remark~\ref{rem:valid-message}, an echo
  message is not handled unless it is signed by the claimed sender.
  More precisely, upon receipt of a message of the form
  $\echo{\mkmsgv{p_i}{\sq}{v}}{\Sigma}$ or
  $\deliv{\mkmsgvs{p_i}{\sq}{v}{\Sigma}}{\Sigma'}$, $p_j$ only handle the
  message if $\Sigma$ contains a signature from $p_i$.  Now, because
  the sender $p_i$ is correct, it must have indeed sent an echo
  message for $\mkmsg{p_i}{v}$.  Finally, we prove by
  induction on the chain of local events happening at $p_i$ (a correct
  process) that led to this message being sent, that $p_i$ must have
  broadcasted it.
\end{proof}

\begin{lemma}[Intersecting delivery]
  \label{lem:intersecting-delivery}
  Let $p$ be a correct process that starts delivering some message $m$
  at some time $t_{d}$.  Then, there exists a collection $B$ of $2f+1$
  processes such that all correct processes
  in $B$
  only deliver $m$ for a full $\mathbb{T}$ duration starting some time
  prior to $t_{d}+\mathbb{T}$.
\end{lemma}

\begin{proof}[Proof outline]
  Let us first point out that because $p$ starts delivering at
  $t_{d}$, and because it is correct, $2f+1$ processes must have
  received this deliver message by $t_{d}+\mathbb{T}$ (otherwise $p$
  would kill itself because it wouldn't be connected---the
  proof-of-connectivity is executed in piggyback mode).  Let $A$ be
  this collection of $2f+1$ processes (note that $p\in{A}$).  For each
  correct process $q\in{A}$, $q$ must have started delivering some
  time prior to $t_{d}+\mathbb{T}$.

  Let us now prove this lemma by induction on $t_{d}$.

  Either a correct process within $A$ started delivering prior to
  $t_d$ or not. If one did, in which case $t_{d}>0$, then we conclude
  by our induction hypothesis.
  Otherwise all correct nodes in $A$ (at least $f+1$) are only
  delivering starting from $t_d$.  Because they start delivering prior
  to $t_d+\mathbb{T}$, and because they deliver for $2\mathbb{T}$, it
  must be that
  all correct processes within that collection
  only deliver $m$ for a full $\mathbb{T}$ duration starting at most
  by $t_d+\mathbb{T}$ (until at most $t_d+2\mathbb{T}$).
\end{proof}

\begin{lemma}[Timely agreement]
  \label{lem:timely-agreement}
  If a correct process $p_i$ broadcasts $m$ at real time $t$, then all
  correct processes deliver $m$ by $t+3\mathbb{T}$.
\end{lemma}

\begin{proof}[Proof outline]
  Since $p_i$ is correct during this broadcast, then
  it must have
  received $2f+1$ echoes for $m$ and must then have started delivering
  $m$ at $t_d\in[t,t+\mathbb{T}]$.
  By Lemma~\ref{lem:intersecting-delivery}, there exists a collection
  $B$ of $2f+1$ processes such that all correct processes
  in $B$ only deliver $m$ for a full $\mathbb{T}$ duration
  starting some time prior to $t_{d}+\mathbb{T}$.
  Now, every other correct process $p_j$ must be connected to $2f+1$
  processes in any proof-of-connectivity period
  $\mathit{pc}=[t_0,t_0+\mathbb{T}]$---let $C(\mathit{pc})$ denote
  those $2f+1$ processes.
  Therefore, because there are $3f+1$ processes, there must be a
  correct process, say $r$, and a proof-of-connectivity period
  $\mathit{pc}=[t_j,t_j+\mathbb{T}]$ at $p_j$ such that: (1)~$r$ is in
  the intersection of $B$ and $C(\mathit{pc})$ (there must be at least
  one correct process in that intersection because it is of size
  $f+1$); and such that
  (2)~$p_j$ received $m$ during $\mathit{pc}$ from $r$, which sent it at
  most by $t_d+2\mathbb{T}$.
  Therefore, $p_j$ must have delivered by  $t+3\mathbb{T}$.
\end{proof}

\begin{lemma}[Agreement]
  \label{lem:agreement-part2}
  If some correct process $p_i$ delivers $m$, then all correct
  processes eventually deliver $m$.
\end{lemma}

\begin{proof}[Proof outline]
  This is a straightforward consequence of
  Lemma~\ref{lem:timely-agreement}.

\end{proof}

\begin{lemma}[Timeliness]
  If a correct process $p_i$ broadcasts $m$ at real time $t$, then no
  correct process delivers $m$ after $t+3\mathbb{T}$.
\end{lemma}

\begin{proof}[Proof outline]
  This is a straightforward consequence of
  Lemma~\ref{lem:timely-agreement}.
\end{proof}

\section{Correctness of PISTIC-CS---Proof of Theorem~\ref{theorem:RTBC-implemented}}
\label{appx:proofRTBC}

Recall that since Algorithm~$\mathcal{A}$ implements interactive
consistency, then when $\mathcal{A}$ eventually terminates all correct
processes will have the same vector of proposals where the values
relative to correct processes are indeed what these correct processes
have proposed. In fact interactive consistency~\cite{Pease:1980}
guarantees the two following properties:
\begin{itemize}
\item[IC.1] The non-faulty processors compute exactly the same vector.
\item[IC.2] The element of this vector corresponding to a given
  non-faulty processor is the private value of that processor.
\end{itemize}

We now prove Thm.~\ref{theorem:RTBC-implemented}, i.e., that assuming
algorithm~$\mathcal{A}$ implements interactive consistency in a known
bounded number of communication rounds (this is used to prove
Lemma~\ref{lemma:RTBC-Timeliness}), as well as
Assumptions~\ref{assump:network-access}
and~\ref{assump:consensus-function}, then~$\mathcal{A}$ implements
RTBC (see Sec.~\ref{sec: RT-consensus}) in our system model (see
Sec.~\ref{Sysmodel}).

\begin{lemma}[RTBC-Termination]
  Every correct process eventually decides.
\end{lemma}

\begin{proof}[Proof outline]
  By Assumption~\ref{assump:network-access}, a correct process $p_i$
  accesses the network only through the RTBRB primitive.  Therefore,
  because $p_i$ is correct and therefore does not enter passive mode
  while executing RTBRB, it must
  terminate.  By IC.1 $p_i$ must compute a vector.  Finally, $p_i$ will
  apply the deterministic function described in
  Assumption~\ref{assump:consensus-function} to that vector to obtain
  a value $v$, which is the value $p_i$ decides upon.

\end{proof}

\begin{lemma}[RTBC-Agreement]
  No two correct processes decide differently.
\end{lemma}

\begin{proof}[Proof outline]
  Let $p_i$ be a correct process that decides upon a value $v_i$, and
  $p_j$ be a correct process that decides upon a value $v_j$. Again,
  by Assumption~\ref{assump:network-access}, $p_i$ and $p_j$ must not
  enter passive mode while using the RTBRB primitive. By IC.1, $p_i$
  and $p_j$ must compute the same vector $V$. Both $p_i$ and $p_j$
  apply the deterministic function described in
  Assumption~\ref{assump:consensus-function} to this vector $V$.
  Therefore, $v_i$ must be equal to $v_j$.
\end{proof}

\begin{lemma}[RTBC-Validity]
  If all correct processes propose the same value $v$, then any
  correct process that decides, decides~$v$.  Otherwise, a correct
  process may only decide a value that was proposed by some correct
  process or the special value $\bot$.
\end{lemma}

\begin{proof}[Proof outline]
  First, note that by Assumption~\ref{assump:network-access}, correct
  processes must not enter passive mode while using the RTBRB primitive.
  Now, if all correct processes propose the same value $v$, then by
  IC.2 the obtained interactive consistency vector computed by a
  correct process should contain $v$ a number of times equal to the
  number of correct processes, i.e., at least $2f+1$ times.  Finally,
  since all correct processes apply the deterministic function
  described in Assumption~\ref{assump:consensus-function} to their
  vectors, they must all decide on $v$.

  Let us now assume that not all correct processes propose the same
  value $v$.  If a correct process $p$ decides upon a value $v'$ then
  by Assumption~\ref{assump:consensus-function}, it must be that
  either (1)~its interactive consistency vector contains at least
  $2f+1$ times this value $v'$; or (2)~that $v'$ is the special value
  $\bot$.  In case $v'$ appears $2f+1$ times in $p$'s interactive
  consistency vector, then by IC.2, it must be that $v'$ was proposed
  by a correct process.
  This concludes the proof.
\end{proof}

\begin{lemma}[RTBC-Timeliness]
  \label{lemma:RTBC-Timeliness}
  If a correct process $p_i$ proposes a value to consensus at time
  $t$, then no correct process decides after $t+\deltaRTBC$.
\end{lemma}

\begin{proof}[Proof outline]
  The way we implement consensus is first by reaching interactive
  consistency and applying a deterministic function after. The
  deterministic function is a computational load that requires
  scanning the consistency vector and hence has a known bounded
  duration since we assume that correct processes are
  synchronous. Therefore, it is sufficient to prove that the
  interactive consistency protocol finishes in a bounded duration (in
  the sense that correct processes compute their interactive
  consistency vectors in a bounded amount of time).

  Recall that we assume that Algorithm~$\mathcal{A}$ requires a
  bounded number of communication rounds to terminate, say~$k$.  By
  Assumption~\ref{assump:network-access} processes send and receive
  messages over the network only via the RTBRB primitive.  Hence any
  communication round has a bounded duration, that being a multiple,
  say $m$, of $\deltaRTBRB$, the duration needed by the RTBRB primitive to
  complete (which is at most $3\mathbb{T}$).
  Therefore, because by Assumption~\ref{assump:network-access},
  correct processes must not enter passive mode while using the RTBRB
  primitive, it must be that correct processes will decide before
  $t+(k\times{m}\times{3\mathbb{T}})$, which concludes our proof.
\end{proof}

\section{Correctness of PISTIC-AT---Proof of Theorem~\ref{theorem:RTBAB}}
\label{sec:pistis-to}

\subsection{Reduction to RTBAB}

To prove Thm.~\ref{theorem:RTBAB}, we only have to prove that
Algorithm~$\mathcal{A}$ satisfies the RTBAB-Timeliness property.

\begin{proof}[Proof outline]
  Let us assume that the correct process $p_i$ RTBAB-broadcasts $m$ at
  time $t$.  We have to prove that no correct process RTBAB-delivers
  $m$ after real time $t+\deltaRTBAB$, for some $\deltaRTBAB$.  We prove
  this by proving the stronger result that there exists a $\deltaRTBAB$
  such that all correct processes RTBAB-deliver $m$ by $t+\deltaRTBAB$.

  By Property~\ref{prop:rtbab-bcast-rtbrb}, $p_i$ RTBRB-broadcasts $m$
  with some sequence number $\RTBABseq_t$ by time $t+\deltaB$.  By
  RTBRB-Validity, RTBRB-Timeliness and RTBRB-Agreement, all correct
  processes RTBRB-deliver $m$ by some time $t+\deltaB+\deltaRTBRB$.
  By Property~\ref{prop:rtbrb-deliv-rtbc-prop}, all correct processes
  will RTBC-propose or RTBC-decide $m$ by
  $t+\deltaB+\deltaRTBRB+\deltaP$.

  If one correct process RTBC-decides $m$ by
  $t+\deltaB+\deltaRTBRB+\deltaP$, then by the RTBC properties,
  all correct processes will RTBC-decide by
  $t+\deltaB+\deltaRTBRB+\deltaP+\deltaRTBC$, and by
  Property~\ref{prop:rtbc-deliv-rtbab}, they will RTBAB-deliver by
  $t+\deltaB+\deltaRTBRB+\deltaP+\deltaRTBC+\deltaD$, which
  concludes our proof.  Therefore, let us now consider the case where
  they all RTBC-propose $m$ by $t+\deltaB+\deltaRTBRB+\deltaP$.

  However, it might be that they RTBC-propose $m$ in different RTBC
  instances.  We want to prove that there will be an RTBC instance
  $\RTBABinst_{m}$ where ``enough'' correct nodes RTBC-propose $m$ at that instance, by time
  $t+\Delta_{m}$ (for some fixed $\Delta_{m}$), so that it results in
  $\RTBABinst_{m}$ deciding $m$.  Then, by RTBC-Termination,
  RTBC-Agreement, RTBC-Timeliness, and
  Property~\ref{prop:rtbc-deliv-rtbab}, we can conclude that all
  correct processes RTBAB-deliver $m$ by time
  $t+\Delta_{m}+\deltaRTBC+\deltaD$.  Let us now prove that
  such an instance $\RTBABinst_{m}$ indeed exists.

  Because all correct processes RTBC-propose $m$ by
  $t+\deltaB+\deltaRTBRB+\deltaP$, there must be a greatest
  instance $\RTBABinst_{g}$ such that a correct process $p_{g}$
  RTBC-proposes $m$ at some time
  $t_k\leq{t+\deltaB+\deltaRTBRB+\deltaP}$.  Now, either (1)~$m$
  was RTBC-decided at a prior instance $\RTBABinst_{p}$ (by all
  correct processes, by the RTBC properties), or (2)~not.  In case it
  was (i.e., case~(1)), all correct processes must have RTBC-decided
  $m$ by time $t+\deltaB+\deltaRTBRB+\deltaP+\deltaRTBC$ by the
  RTBC properties and because $\RTBABinst_{p}$ must have been dealt
  with by $p_{g}$ before $\RTBABinst_{g}$ by
  Property~\ref{prop:rtbc-unique}.  Now, by
  Property~\ref{prop:rtbc-deliv-rtbab}, it must be that all correct
  processes must have RTBAB-delivered $m$ by time
  $t+\deltaB+\deltaRTBRB+\deltaP+\deltaRTBC$.

  Let us now focus on case~(2), i.e., $m$ was not RTBC-decided at a
  prior instance.  By Property~\ref{prop:rtbc-prop-val-or-bot}, correct
  processes must be RTBC-proposing either $m$ or $\bot$ at instance
  $\RTBABinst_{g}$.  Let us prove that they cannot propose $\bot$, in
  which case we conclude using RTBC-Validity and
  Property~\ref{prop:rtbc-deliv-rtbab}, and $\deltaRTBAB$ is again
  $t+\deltaB+\deltaRTBRB+\deltaP+\deltaRTBC$.
  We prove that correct processes cannot propose $\bot$ at instance
  $\RTBABinst_{g}$ by contradiction.  Let us assume that some correct
  process $p_{j}$ votes for $\bot$ at instance $\RTBABinst_{g}$
  (therefore, $p_{j}$ cannot be $p_{g}$).
  By definition of $\RTBABinst_{g}$, it must be that $p_j$ votes for
  $m$ at a prior instance $\RTBABinst_{p}$.  Because it is an instance
  prior to $\RTBABinst_{p}$, as mentioned above, $m$ was not
  RTBC-decided at that instance.  Therefore, by
  Property~\ref{prop:rtbc-prop-val-or-bot}, and RTBC-Validity, it must be
  that this instance ended up in $\bot$ being decided.
  Finally, we obtain a contradiction from the fact that $p_j$ must
  also RTBC-propose $m$ at instance $\RTBABinst_{g}$, which we prove
  by induction on the list of instances between $\RTBABinst_{p}$ and
  $\RTBABinst_{g}$ and using Property~\ref{prop:rtbc-repropose}.

\end{proof}

\subsection{PISTIS-AT: a Class of Algorithms Implementing RTBAB}
\label{sec:pistis-to-algo}

Algorithm~\ref{alg:RTBAB} provides an example of a PISTIS-AT
algorithm, which implements the RTBAB primitive presented in
Sec.~\ref{sec:atomic-bcast}.
We assume here that a process broadcasts a message by invoking
$\rtbabBcastM$, and delivers a message
invoking $\rtbabDelivM$.
In addition, $\rtbabInit{\mbox{rtbab}}$ instantiates a new instance of
RTBAB with id $\mbox{rtbab}$.
To guarantee total order, each process maintains a monotonically
increasing sequence number $\META{seq}$, which is incremented
every time $\rtbabBcastM$ is called.

\begin{lemma}
  \label{lem:pistis-to-rounds}
  Given an RTBAB instance $\RTBABinst$, such that $p_i$ is the leader
  of $\RTBABinst$, all correct processes will either RTBC-propose a
  value received from $p_i$ or $\bot$ (in case they have not received
  any new message from $p_i$ since the last one they processed).
  Moreover, given two correct processes that RTBC-propose such values
  at instance $\RTBABinst$, it must be that either those values are
  equal (to the $k^{th}$ new value broadcasted by~$p_i$, for some~$k$)
  or one of them is $\bot$ (in case the corresponding process has not
  received $p_i$'s $k^{th}$ broadcasted new value yet, and has already
  processed all previous broadcasted value from~$p_i$).
\end{lemma}

\begin{proof}[Proof outline]
  This can be proved by induction on causal time.

  The first time those correct processes RTBC-propose a value at an
  instance such that $p_i$ is the leader, it must be that either this
  value is the first value RTBAB-broadcasted by $p_i$, or $\bot$.

  The inductive case goes as follows: we assume that our property is
  true at a given instance $\RTBABinst$ such that $p_i$ is the leader,
  and where correct processes RTBC-propose either $v$ (the
  $(k-1)^{th}$ new value proposed by $p_i$) or $\bot$, and we prove
  that the property is still true at the next such instance
  $\RTBABinst'$.  By RTBC-Validity, it must be that correct processes
  either RTBC-decide $v$ or $\bot$, and by RTBC-Agreement, they must
  not decide differently.
  Therefore, if they decide $v$ at instance $\RTBABinst$, then $v$
  will be added to the $\RTBABdelivered$ set, and therefore never
  added to $\RTBABunordered$ again; and in addition, it will be
  removed from $\RTBABunordered$.  At the next instance $\RTBABinst'$,
  these processes will vote either for the $k^{th}$ new value proposed by
  $p_i$ or for $\bot$ if they have not received that $k^{th}$ new value.
  In particular, if one of those correct processes RTBC-proposed
  $\bot$ because it had not received $v$ yet, then at instance
  $\RTBABinst'$ it will either propose the $k^{th}$ new value proposed by
  $p_i$ (since $v$ is skipped because already delivered), or $\bot$ in
  case it has not received this $k^{th}$ new value yet.
  Otherwise if they decide $\bot$, then the correct processes that
  voted for $v$ will still vote for $v$ at $\RTBABinst'$, and those
  that voted for $\bot$ will either keep on voting for $\bot$ if they
  still have not received $v$, or finally receive $v$ and start voting
  for $v$.  Note that by RTBAB-Agreement, all correct processes must
  eventually receive $v$.
\end{proof}

In order to obtain time bounds that do not depend on
Algorithm~\ref{alg:RTBAB}'s variable, we make the following
assumption:
\begin{assumption}
  \label{assump:bcast-one-at-a-time}
  Correct processes wait for
  $\deltaRTBRB+\deltaWait+({n}\times(\deltaRTBC+\deltaWait))$ between
  two different broadcasts.
\end{assumption}
As we will see below, this is the time it takes to guarantee that all
correct processes RTBAB-deliver an RTBRB-broadcasted value.

\begin{lemma}
  Algorithm~\ref{alg:RTBAB} satisfies Property~\ref{prop:rtbab-bcast-rtbrb}.
\end{lemma}

\begin{proof}[Proof outline]
  Property~\ref{prop:rtbab-bcast-rtbrb} holds because Algorithm~\ref{alg:RTBAB}
  RTBTB-broadcasts messages on each call to $\rtbabBcastSYMB$ (see
  lines~\ref{pistis-to-event} and~\ref{pistis-to-rtbrb-bcast} of
  Algorithm~\ref{alg:RTBAB}).
\end{proof}

\begin{lemma}
  Algorithm~\ref{alg:RTBAB} satisfies Property~\ref{prop:rtbc-deliv-rtbab}.
\end{lemma}

\begin{proof}[Proof outline]
  If a value $v$ (different from $\bot)$ is RTBC-decided at time $t$, and
  the RTBC instance is the current instance, and $v$ is not in
  $\RTBABdelivered$, then it is RTBAB-delivered.  If $v$ is in
  $\RTBABdelivered$, then it must be that it was added to that set in
  the past, in which case it was delivered at that time.

  Now, if the RTBC instance is not the current instance, Algorithm~\ref{alg:RTBAB}
  retries handling the messages after a while.  The number of times a
  process will retry handling deliver messages is bounded because
  instances are handled in a monotonic order and are bounded in time
  according to RTBC-Timeliness.
\end{proof}

\begin{lemma}
  Under Assumption~\ref{assump:bcast-one-at-a-time}, Algorithm~\ref{alg:RTBAB}
  satisfies Property~\ref{prop:rtbrb-deliv-rtbc-prop}.
\end{lemma}

\begin{proof}[Proof outline]
  First of all, let us point out that RTBRB-deliver messages are
  treated in monotonic order.  Let us now consider three cases.  In
  the following, we first provide variable-dependent bounds, and we
  then explain how to get independent bounds using
  Assumption~\ref{assump:bcast-one-at-a-time}.

  Case~(1): Whenever a process $p_i$ receives a RTBRB-deliver message
  $m$ at time $t$ with sequence number $\META{num}$, broadcasted by
  $p_j$, which is the next one to receive (i.e.,
  $\META{num}=\RTBABnext[p_j]$), and if $m$ is not already in
  $\RTBABdelivered$, then $p_i$ will append $m$ to its
  $\RTBABunordered[p_j]$ list.  We now have to prove that $m$ will
  then be RTBC-proposed or RTBC-decided by some time $t+\deltaP$,
  for some bounded $\deltaP$.  Because $m$ is now in $p_i$'s
  $\RTBABunordered[p_j]$ list, the event
  line~\ref{alg:exists-unordered} will be triggered at least until $m$
  is removed from the list.  Because Algorithm~\ref{alg:RTBAB} uses the rotating
  coordinator paradigm, then a value broadcasted by some process $p_k$
  is voted upon using an RTBC instance only every $n$ (the total
  number of processes) instances (i.e., whenever $p_k$ is the leader).
  However, there might be other values before $m$ in the
  $\RTBABunordered[p_j]$ lists maintained by the processes.  The
  processes have to RTBC-decide these previous values to start
  RTBC-proposing $m$ if $m$ has not been RTBAB-delivered in the
  meantime (otherwise we can conclude because RTBAB-delivered messages
  are RTBC-decided upon).  Because of the rotating coordinator scheme,
  and by the RTBC properties and Lemma~\ref{lem:pistis-to-rounds}, we
  get the guarantee that $m$ will be RTBC-proposed by
  $t+(n\times(\deltaRTBC+\deltaWait)\times(\META{num}+1))$, where
  $\deltaRTBC+\deltaWait$ is the time it takes to complete an RTBC
  instance, and $n\times(\deltaRTBC+\deltaWait)$ is the time it takes
  to rotate through the leaders ($\deltaWait$ is the time processes
  wait for before re-trying to handle a message---see
  line~\ref{alg:wait-deliv} and line~\ref{alg:wait-decide}).  Now,
  thanks to Assumption~\ref{assump:bcast-one-at-a-time}, we can derive
  that all previous values stored in $\RTBABunordered[p_j]$ have
  already been decided upon when correct processes deliver $m$.
  Therefore, we get that $m$ will be RTBC-proposed by
  $t+(n\times(\deltaRTBC+\deltaWait))$.

  Case~(2): If $m$ is already in $\RTBABdelivered$, then $p_i$ must have
  already RTBC-decided $m$ according to
  lines~\ref{alg:rtbc-decide}--\ref{alg:add-to-delivered}.

  Case~(3): If $m$ is not the next value that $p_i$ is supposed to
  receive, it will re-try RTBRB-delivering $m$ after $\deltaWait$
  until it has received all the previous values.  The RTBRB properties
  guarantee that if some correct process $p_j$ broadcasts a value $v$
  at time $t$, then correct processes will deliver $v$ by
  $t+\deltaRTBRB+\deltaWait$.  Therefore, it must be that correct
  processes will have stored $m$ (and all previous values) in their
  $\RTBABunordered[p_j]$ list by
  $t+(\deltaRTBRB+\deltaWait)\times(\META{num}+1)$.  Finally, following the
  same argument as above, we get that $m$ will be RTBC-proposed by
  $t+((\deltaRTBRB+\deltaWait)\times(\META{num}+1))+(n\times(\deltaRTBC+\deltaWait)\times(\META{num}+1))$.
  As mentioned above,
  thanks to Assumption~\ref{assump:bcast-one-at-a-time}, we can derive
  that all previous values stored in $\RTBABunordered[p_j]$ have
  already been RTBRB-delivered and RTBC-decided upon when correct
  processes deliver $m$.
  Therefore, we get that $m$ will be RTBC-proposed by
  $t+((\deltaRTBRB+\deltaWait))+(n\times(\deltaRTBC+\deltaWait))$.
\end{proof}

\begin{lemma}
  Algorithm~\ref{alg:RTBAB} satisfies Property~\ref{prop:rtbc-prop-val-or-bot}.
\end{lemma}

\begin{proof}[Proof outline]
  This is a straightforward consequence of
  Lemma~\ref{lem:pistis-to-rounds}.
\end{proof}

\begin{lemma}
  Algorithm~\ref{alg:RTBAB} satisfies Property~\ref{prop:rtbc-repropose}.
\end{lemma}

\begin{proof}[Proof outline]
  Let $p_i$ be a correct process that proposes a value $v$, with
  broadcaster $p_j$, at a given time $t$, using a given RTBC instance
  $\RTBABinst$, and such that this instance does not decide $v$.  By
  Lemma~\ref{lem:pistis-to-rounds}, all correct processes propose $v$
  or $\bot$ at that instance.  By the RTBC properties, because
  $\RTBABinst$ does not decide $v$, it must decide $\bot$.  Therefore,
  $p_i$ will increment its RTBC instance number but will keep $m$ at
  the head of its $\RTBABunordered[p_j]$ list.  After a full rotation
  through the leaders, it will RTBC-propose $v$ again at the later
  instance $\RTBABinst+n$, where $0<n$.

  Moreover, no correct process will propose $v$ between $\RTBABinst$
  and $\RTBABinst+n$ because $p_j$ ($v$'s broadcaster) is the leader
  of $\RTBABinst$ and $\RTBABinst+n$ but not of the instances in
  between, and $v$ can only be in the $\RTBABunordered[p_j]$ lists.
\end{proof}

\begin{lemma}
  Algorithm~\ref{alg:RTBAB} satisfies Property~\ref{prop:rtbc-unique}.
\end{lemma}

\begin{proof}[Proof outline]
  By design, correct processes RTBC-propose exactly one value per RTBC
  instance because they only start proposing a value in a new instance
  if $\RTBABbusy$ is $\FALSE$; in which case they set $\RTBABbusy$ to
  $\TRUE$; wait for this instance to complete; and finally increment
  the RTBC instance number and set back $\RTBABbusy$ to $\FALSE$.

  Correct processes propose values in all RTBC instances and
  monotonically because they increment the the RTBC instance number by
  one every time an RTBC instance complete.

  Finally, correct processes do not run RTBC instances in parallel
  thanks to the $\RTBABbusy$ flag.
\end{proof}

\subsection{Direct Proof of Algorithm~\ref{alg:RTBAB}'s Correctness}
\label{sec:proof-RTBAB}
\label{sec:pistis-to-proofs}

\begin{lemma}[RTBAB-Validity]
  If a correct $p_i$ process broadcasts $m$, then $p_i$ eventually
  delivers $m$.
\end{lemma}

\begin{proof}[Proof outline]
  By RTBRB-Validity, $p_i$ eventually delivers $m$ with sequence
  number $\META{num}$.  If $p_i$ has already delivered $m$, i.e.,
  $m\in\RTBABdelivered$, then we are done.  Otherwise, because $p_i$
  broadcasts messages monotonically (and without gaps), it will append
  $m$ to its list of unordered messages
  (line~\ref{alg:append-new-message} of Algorithm~\ref{alg:RTBAB}).
  Therefore, line~\ref{alg:exists-unordered} will be triggered until
  $m$ is removed from the list, as long as $p_i$ eventually resets
  $\RTBABbusy$ to $\FALSE$ once it has set it to $\TRUE$, which is
  true by RTBC-termination.  When finally $p_i$ is the leader of its
  current instance, say $\RTBABinst_1$, and that $m$ is at the head of
  $p_i$'s unordered list, $p_i$ will RTBC-propose $m$.  By
  RTBC-Validity, either all the correct processes RTBC-propose $m$, in
  which case $p_i$ delivers $m$; or some correct processes
  RTBC-propose values different from $m$.  As mentioned above, such
  proposed values must then be $\bot$, in which case $p_i$ might
  RTBC-decide $m$ or $\bot$.  Again as mentioned above, if $p_i$ does
  not deliver $m$, it will again either decide $m$ or $\bot$ at the
  next instance where it is the leader.  Because by RTBAB-Agreement,
  all correct processes eventually receive $m$, it must be that
  eventually, $p_i$ RTBC-decides $m$ for an instance where it is the
  leader, and in turn RTBAB-deliver $m$.
\end{proof}

\begin{lemma}[RTBAB-No duplication]
  No message is delivered more than once.
\end{lemma}

\begin{proof}[Proof outline]
  This property straightforwardly follows trivially from the fact that
  delivered values are added to the $\RTBABdelivered$ set
  line~\ref{alg:add-to-delivered}, and from the fact that a process
  always checks whether it has delivered a message $m$ before
  delivering $m$ (see line~\ref{alg:check-if-delivered}).
\end{proof}

\begin{lemma}[RTBAB-Integrity]
  If some correct process delivers a message $m$ with initial sender $p_i$ and
  process $p_i$ is correct, then $m$ was previously broadcast by
  $p_i$.
\end{lemma}

\begin{proof}[Proof outline]
  First of all, the RTBAB-delivered value $m$ (which must be different
  from $\bot$) with sender $p_i$ (i.e., such that $p_i$ is the leader
  of the current instance) must have been RTBC-decided upon.  By
  RTBC-Agreement and RTBC-Termination, it must be that the correct
  sender $p_i$ has also RTBC-decided upon $m$.  It must be that $m$
  was it $p_i$'s own $\RTBABunordered$ list.  Therefore, it must be
  that $p_i$ RTBRB-delivered $m$.  Finally, by RTBRB-Integrity, it
  must be that $p_i$ previously broadcasted $m$.
\end{proof}

\begin{lemma}[RTBAB-Agreement]
  If some message $m$ is delivered by any correct process, then every
  correct process eventually delivers $m$.
\end{lemma}

\begin{proof}[Proof outline]
  Let $p_i$ be the process that RTBAB-delivered $m$ at instance
  $\RTBABinst$, such that $p_l$ is the leader of that instance.  This
  delivered value must be different from $\bot$, and must have been
  RTBC-decided upon.  By RTBC-Agreement and RTBC-Termination, it must
  be that all correct processes eventually RTBC-decide $m$ as well.
  Let $p_j$ be one such correct process.  We have to prove that $p_j$
  RTBAB-delivers $m$ also at instance $\RTBABinst$.  By
  RTBRB-Agreement, it must be that $p_j$ eventually receives the same
  broadcasts as $p_i$, among other things, those for which $p_l$ is
  the leader.  From RTBC-Agreement and RTBC-Termination, it must be
  that all correct processes eventually decide the same values for
  each RTBC instance.  Therefore, $p_j$ will eventually reach instance
  $\RTBABinst$, and will therefore also RTBAB-deliver $m$.
\end{proof}

\begin{lemma}[variable-dependent RTBAB-Timeliness]
  \label{lem:pre-RTBAB-Timeliness}
  There exists a known $\deltaRTBAB$ such that if a correct process
  $p_i$ broadcasts $m$ at time $t$, no correct process delivers $m$
  after real~time~$t+\deltaRTBAB$, where $\deltaRTBAB$ depends on
  $\RTBABseq_{t}$, the current sequence number at the time $m$ is
  broadcasted.
\end{lemma}

\begin{proof}[Proof outline]
  Timeliness follows from RTBRB-Timeliness and RTBC-Timeliness, as
  well as of the fact that Algorithm~\ref{alg:RTBAB} rotates through
  the processes (processes might have to wait a full rotation before
  they get a chance to decide on a messages that was
  RTBAB-broadcasted).
  Let $\deltaRTBC$ be the time it takes for all correct processes to
  decide on a value using RTBC (see RTBC-Timeliness).
  Let $\deltaRTBRB$ be the time it takes for all correct processes to
  deliver a message using RTBRB (which exists by RTBRB-Timeliness).
  Assume that $p_i$ assigns the sequence number $\RTBABseq_t$ with
  the message $m$.
  As mentioned above, we assume that $p_i$ RTBAB-broadcasts $m$ at
  time $t$.
  Because correct processes might still be RTBRB-delivering messages
  when they gets the RTBRB-deliver message for $m$, they might not be
  able to RTBRB-deliver $m$ right away (it might be that
  $\RTBABseq_t>\RTBABnext[p_i]$).  However, we are guaranteed that all
  correct processes will have delivered $m$ by time
  $T_1=t+((\deltaRTBRB+\deltaWait)\times(\RTBABseq_t+1))$ (where
  $\deltaWait$ is the time processes wait for before re-trying to
  handle a message---see line~\ref{alg:wait-deliv} and
  line~\ref{alg:wait-decide}).
  Note that at that time, processes might be RTBAB-delivering other
  messages broadcasted by other processes than $p_i$.  Also, there
  might already be some messages from $p_i$ to RTBAB-deliver before
  $m$ (all those with sequence numbers less than $\RTBABseq_t$).  In
  case $p_i$ is currently not the leader, it might have to
  wait a full rotation through the processes to get a chance to be the
  leader again.  Given the fact that all correct processes have $m$ in
  their $\RTBABunordered$ list by time $T_1$, a full rotation will take at most
  ${n}\times(\deltaRTBC+\deltaWait)$.
  Because processes might have to process $\RTBABseq_t$ messages from $p_i$
  before they get a chance to process $m$, it follows that $m$ will be
  RTBAB-delivered by
  $t+((\deltaRTBRB+\deltaWait)\times(\RTBABseq_t+1))+({n}\times(\deltaRTBC+\deltaWait)\times(\RTBABseq_t+1))$.
\end{proof}

As mentioned in Def.~\ref{def:rtbab}, the RTBAB timeliness bound is
different from the RTBRB one. $\deltaRTBAB$ is the RTBAB bound, while
$\deltaRTBRB$ is the RTBRB bound.

\begin{lemma}[RTBAB-Timeliness]
  Under Assumption~\ref{assump:bcast-one-at-a-time}, there exists a
  known $\deltaRTBAB$ such that if a correct process $p_i$ broadcasts
  $m$ at time $t$, no correct process delivers $m$ after
  real~time~$t+\deltaRTBAB$.
\end{lemma}

\begin{proof}[Proof outline]
  Using Assumption~\ref{assump:bcast-one-at-a-time} and a proof
  similar to the one of Lemma~\ref{lem:pre-RTBAB-Timeliness}, we
  derive that messages RTBAB-broadcasted at time $t$ are
  RTBAB-delivered by
  $t+(\deltaRTBRB+\deltaWait+({n}\times(\deltaRTBC+\deltaWait)))$.
\end{proof}

As mentioned in Def.~\ref{def:rtbab}, in addition to the RTBRB
properties, RTBAB also include a total order property.

\begin{lemma}[RTBAB-Total order]
  Let $m_1$ and $m_2$ be any two messages and suppose that $p_i$ and
  $p_j$ are any two correct processes that deliver $m_1$ and $m_2$. If
  $p_i$ delivers $m_1$ before $m_2$, then $p_j$ delivers $m_1$
  before~$m_2$.
\end{lemma}

\begin{proof}[Proof outline]
  Because $p_i$ RTBAB-delivers $m_1$ before $m_2$, it must have
  RTBC-decided $m_1$ at an instance $\RTBABinst_1$ and $m_2$ at an
  instance $\RTBABinst_2$ such that $\RTBABinst_1<\RTBABinst_2$.  By
  RTBC-Agreement and RTBC-Termination, $p_j$ must also have
  RTBC-decided $m_1$ at $\RTBABinst_1$ and $m_2$ at $\RTBABinst_2$.
  Using a similar argument as in the proof of RTBAB-Agreement, we
  derive that $p_j$ must then also have RTBAB-delivered $m_1$ at
  instance $\RTBABinst_1$ and $m_2$ at $\RTBABinst_2$.
\end{proof}

\section{Evaluation Using Number of Messages Sent}
\label{sec:appendixmsgs}

\begin{figure}[h]
  \begin{center}
    \includegraphics[width=0.40\textwidth]{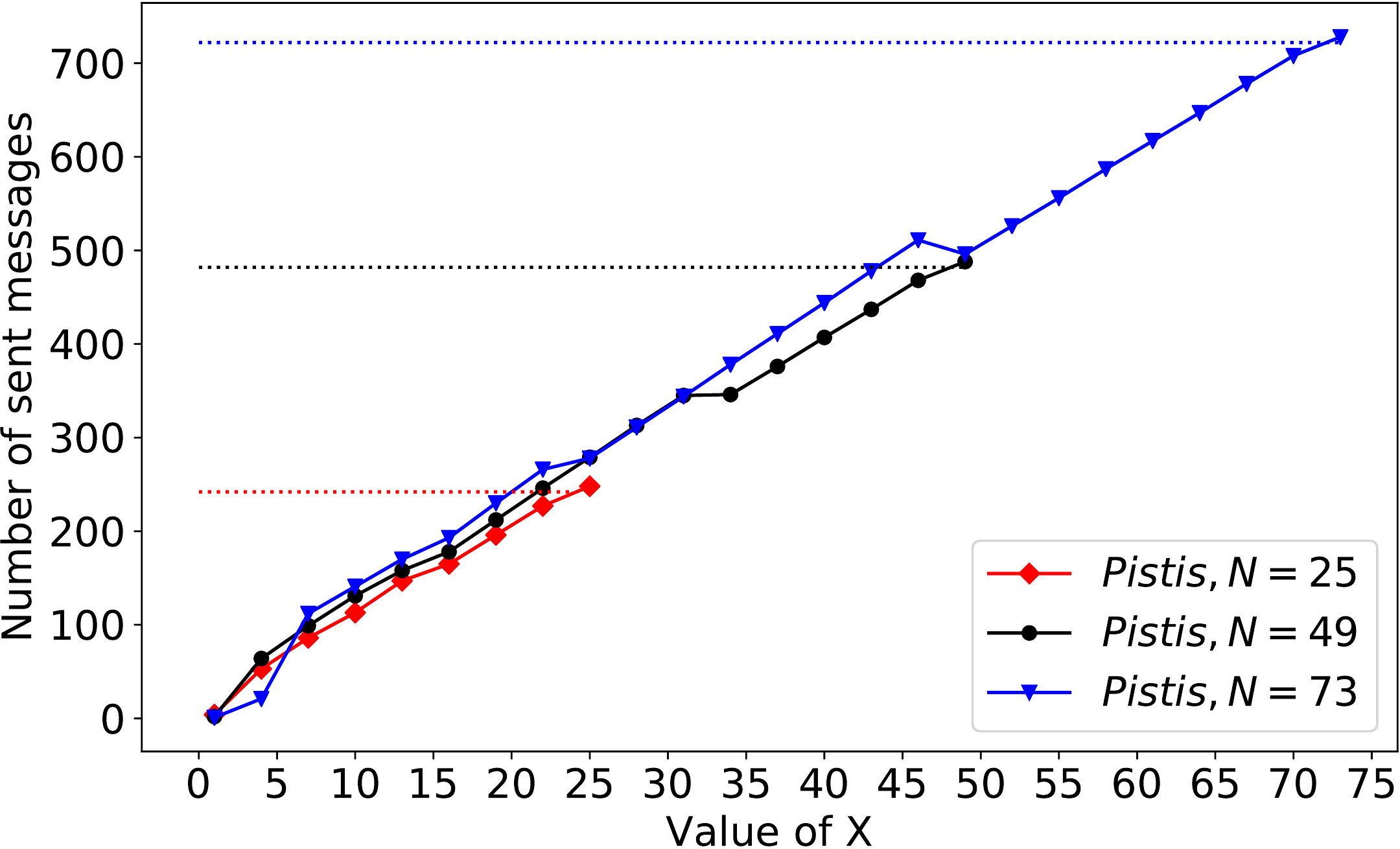}
  \end{center}
  \vspace*{-10pt}
  \caption{Average number of messages transmitted per node with a 1ms link latency, with system sizes equal to 25, 49 and 73 for Pistis and RT-ByzCast.}
  \label{fig:msgs}
\end{figure}

To complement the bandwidth consumption evaluation that was previously
reported, Fig.~\ref{fig:msgs} presents the number of messages
transmitted using either Pistis or RT-ByzCast. We considered systems
containing 25, 49 and 73 nodes (i.e., 3f+1 for f equals to 8, 16 and
24). We used a 1ms network latency and 1B messages. RT-ByzCast's
values are reported with dashed horizontal lines. One can see that
Pistis sends less messages when the value of $X$ decreases. In
addition, PISTIS always sends less messages than RT-ByzCast. In
particular, PISTIS and RT-ByzCast approximately send the same number
of messages when $X=3f+1$. These results are consistent with the
bandwidth consumption results reported in
Sec.~\ref{sec:PISTIS-overhead}, and which therefore indicate that the
main reason behind Pistis' lower bandwidth consumption is a smaller
number of messages exchanged.

\printbibliography[heading=subbibliography]
\end{refsection}

\else%
\fi

\end{document}